\documentclass{acmsmall}

\usepackage[linesnumbered,ruled]{algorithm2e}

\SetAlFnt{\small}
\SetAlCapFnt{\small}
\SetAlCapNameFnt{\small}
\SetAlCapHSkip{0pt}
\IncMargin{-\parindent}

\setcopyright{none}

\usepackage{amsmath,amssymb}
\usepackage{subfigure}
\usepackage{graphicx}
\usepackage{ulem}
\usepackage{xspace}
\usepackage{setspace}
\usepackage{stackrel}
\usepackage{verbatim}
\usepackage{algpseudocode}
\usepackage{color}
\usepackage{ifthen}
\usepackage{comment}
\usepackage{graphicx}
\usepackage{subfigure}
\usepackage{enumerate}
\usepackage{mathrsfs}
\usepackage{picture}
\usepackage[colorlinks,linkcolor=blue,anchorcolor=blue,citecolor=black,bookmarksopen=true]{hyperref}
	\newcommand{\subparagraph}{}
 \usepackage{titlesec}
\usepackage{overpic}
\usepackage{multirow}
\usepackage{booktabs}
\usepackage[all]{xy}
\usepackage{algorithm2e}
\usepackage[square,semicolon,sort]{natbib}

\allowdisplaybreaks
\def \fp {\textit{fp}}
\def \wss {\textit{wss}}

\def \rt {\textit{rt}}
\def \mr {\textit{mr}}

\def \ft {\textit{ft}}
\def \im {\textit{im}}
\def \lt {\textit{la}}

\def \lte {l_e}
\def \fte {f_e}
\def \ltea {l_e'}
\def \ftea {f_e'}

\def \fthead {f'}
\def \lthead {l'}
\def \rthead {rt'}

\def \fttail {f''}
\def \lttail {l''}
\def \rttail {rt''}

\def \ftehead {f_e'}
\def \ltehead {l_e'}
\def \ftetail {f_e''}
\def \ltetail {l_e''}

\def \ai {\textit{ai}}
\def \pd {\textit{pd}}
\def \mt {\textit{mt}}

\begin{document}

\normalem

	\setcounter{tocdepth}{3}
	\setcounter{secnumdepth}{3}

	\newcommand {\myd}{\;\ mathrm{d}}
	\newcommand{\algorithmicname}{\textbf{}}
	\newcommand{\bigo}{O}
	\newcommand{\TO}{\textbf{\ to\ }}
	\newcommand{\STEP}{\textbf{\ step\ }}

\setlength{\pdfpageheight}{\paperheight}
\setlength{\pdfpagewidth}{\paperwidth}

\markboth{L. Yuan et al.}{A Measurement Theory of Locality }

\author{LIANG YUAN
\affil{SKL of Computer Architecture, Institute of Computing Technology, CAS}
CHEN DING
\affil{University of Rochester}
PETER DENNING
\affil{Naval Postgraduate School}
YUNQUAN ZHANG
\affil{SKL of Computer Architecture, Institute of Computing Technology, CAS}}

\title{A Measurement Theory of Locality (MTL)}

\begin{abstract}

Locality is a fundamental principle used extensively in program and system optimization.  It can be measured in many ways.  This paper formalizes the metrics of locality into a measurement theory.  The new theory includes the precise definition of locality metrics based on access frequency, reuse time, reuse distance, working set, footprint, and the cache miss ratio.  It gives the formal relation between these definitions and the proofs of equivalence or non-equivalence.  It provides the theoretical justification for four successful locality models in operating systems, programming languages, and computer architectures which were developed empirically.


\end{abstract}




\begin{bottomstuff}
The manuscript is new and not a revision of a previous conference paper.

\end{bottomstuff}

\maketitle

\section{Introduction}

Locality is a fundamental property of computation and a central principle in software, hardware and algorithmic design~\citep{Denning:CACM05loc}.  As defined by Denning, it is the ``tendency for programs to cluster references to subsets of address space for extended periods.''~\citep[pp. 143]{DenningM:Book15}

Existing literature provides many ways to measure locality: reuse frequency, miss
frequency, reuse distance, footprint and working set.  They are intuitively related, i.e. data reuse in a
program is likely to become data reuse in cache and therefore reduces
the miss frequency.  The relation, however, is not all
clear.  Some metrics, e.g. reuse distance, do not depend on cache size, but other metrics, e.g. miss ratio, do. 
Without a precise relation, we do not know which data reuse becomes a cache hit.  As
a result we do not have reliable properties.  For example, it is
possible to have more reuses in a program but fewer reuses in cache.
Locality optimization cannot be sufficiently formulated without knowing how optimizing for one metric would affect other metrics. 

In this paper, we give a measurement theory of locality (MTL) to formalize the relation between a set of locality metrics.  Measurement of locality is the assignment of a number so that locality of different programs can be compared.  The measurement theory consists of a set of locality metrics, the relation between them, and their precision and error. 




The theory has a limited scope.  It is a theory about locality measurements
but not directly about locality optimization.  It can compute
the amount of data transfer in a memory hierarchy but does not
minimize the amount of data transfer, nor does it optimize the
data layout, which is a more complex problem (for either processor
cache~\cite{PetrankR:POPL02} or virtual memory~\cite{Lavaee:POPL16}).
It assumes automatic cache management by least-recently used (LRU)
replacement or similar policies.  It does not solve the more general
problem of I/O complexity~\cite{HongK:STOC81,Elango+:POPL15}.



\section{Measurement Theory of Locality}

We will first present an overview that divides the locality measurements 
into six categories and then present them in subsections.

\subsection{Overview}

Figure~\ref{fig:ct-tree} shows locality metrics in three top-level and four second-level categories.  At the top level, \emph{access metrics} are measures of locality for each memory access.  The other two types are mathematical functions: \emph{timescale metrics}, whose parameter is a length of time, and \emph{cache metrics}, whose parameter is a cache size.  For these metrics, the measurement theory gives their precise definitions and properties.

\begin{figure}[h]
\centering
\includegraphics[width=8cm]{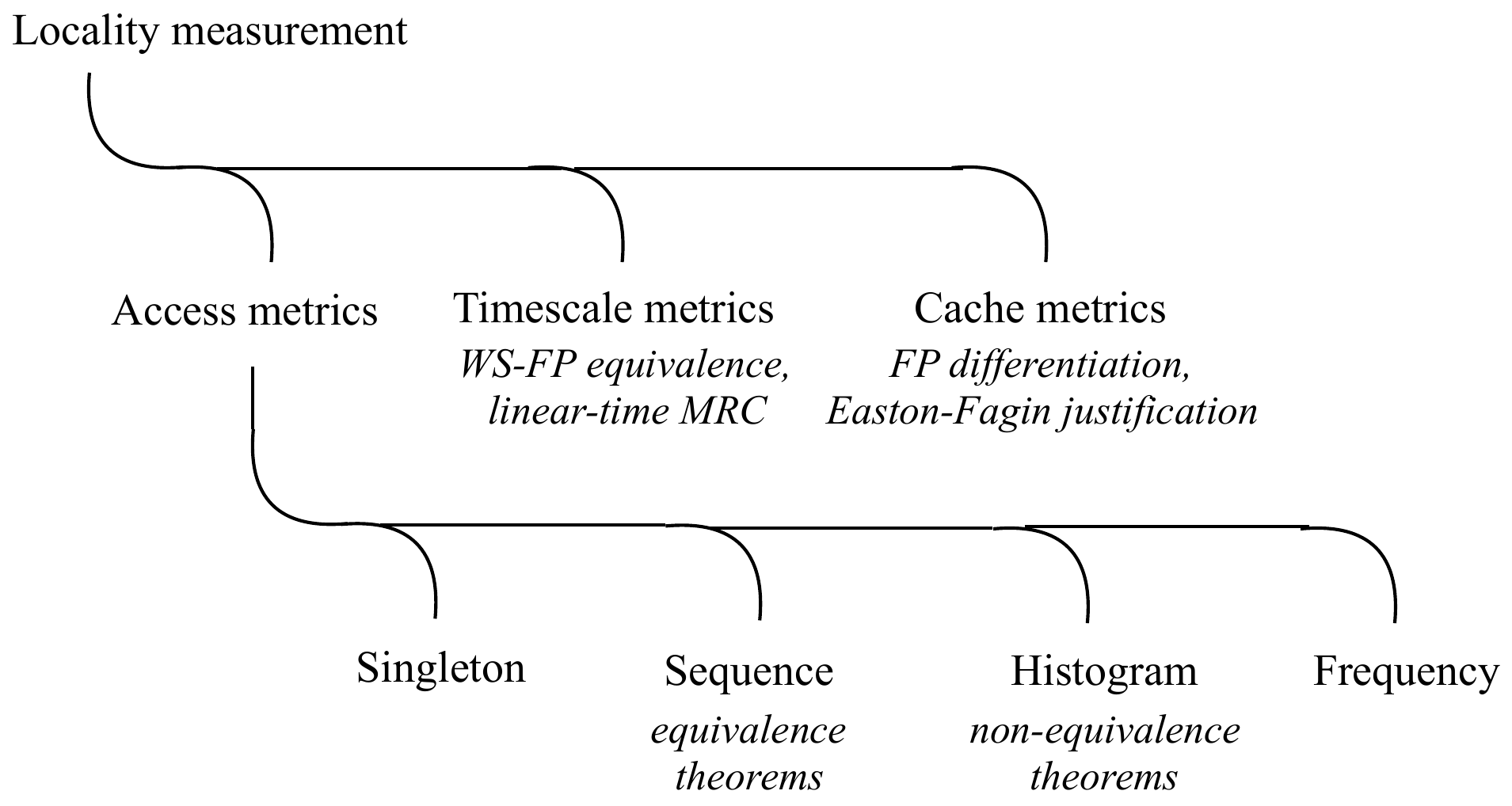}
\caption{The categories of locality measurements and new theoretical results (in italics) }
\label{fig:ct-tree}
\end{figure}

First, the paper formalizes access metrics into four second-level sub-categories based on memory addresses, access times and data reuses.  Then data reuses may be quantified by time or distance and organized as sequences (time ordered) or histograms (sorted by magnitude).  A set of properties of equivalence and non-equivalence are formally derived, including a constructive algorithm to show that the order of reuses of each data item implies the order of all reuses. These definitions and properties are used to show the strength and weakness of reuse distance histograms, the most abstract, compact and widely used metric of access locality.




Second, the paper shows the equivalence between two timescale metrics: the working-set by~\citet{DenningS:CACM72} and the footprint by ~\citet{Xiang+:PACT11,Xiang+:ASPLOS13}.  While access metrics cannot fully model program interaction in shared cache, timescale metrics can.  The mathematical properties of timescale metrics, boundedness and concavity, are trivial consequences of this equivalence.  In addition, the paper gives a much simplified explanation of the footprint formula invented by \citet{Xiang+:PACT11}.


Third, in cache metrics, the paper computes the miss ratio as the derivative of footprint~\cite{Xiang+:PACT11,Xiang+:ASPLOS13}.  Based on this footprint differentiation, it explains a formula by \citet{EastonF:CACM78} almost four decades ago, which computed the miss ratio of a larger cache from the miss ratio of all smaller caches. In addition, it gives an algorithm that is asymptotically faster than \citet{Xiang+:PACT11} for using the footprint to compute the miss ratio.

Finally, we summarize with a conversion theory that connects all the metrics in these categories. The theorems in this paper provide the pieces missing from previous work but necessary to show complete relationship. 

\subsection{Access Metrics}

An execution trace is a sequence of data items referred to by the complete execution of a program.  Each data item is represented by its memory address. The words ``sequence", ``trace" and ``execution" are used interchangeably, so are the phrases ``memory access" and ``memory address".  Hence, an execution trace is the same as a memory address sequence. We ignore any issue of granularity.  A data item may be a variable, cache block, page or object.  We define the following:

\begin{itemize}
\itemindent 20pt
\item[] $N=\mt(1\dots n)$ is a memory address trace of length $n$.
\item[] $M=\{e_1 \dots e_m \}$ is the set of $m$ data items, i.e. distinct memory addresses, accessed by the trace $N$.
\end{itemize}

This section describes three categories of access metrics, singleton, sequence and histogram, and leaves the fourth category, frequency, to Section~\ref{sec:freq}.



\subsubsection{Singleton Locality}

The simplest measurement is no measurement.  The \emph{singleton metric} ``measures" the locality of an access by the access itself.  The locality of an execution trace is the trace itself.  The name ``singleton'' is an adaptation of Lu and Scott in their formalization of determinism, which divides executions of a concurrent program into equivalence classes~\cite{LuS:DISC11}.  Like singleton determinism, singleton locality is the strictest definition of equivalence.  Two executions have the same locality if and only if they are identical. Other metrics are less restrictive, i.e. more abstract and higher level, which means a coarser partition of execution traces into equivalence sets. 


\subsubsection{Sequence Metrics}
\label{sec:seq}

The locality of an access trace may be measured by one of the following three sequences:

\begin{itemize}
\item \emph{Address independence (AI).}  The metric is a transformation of an access trace, by renaming the memory addresses to $M=\{ 1 \dots m \}$ and assigning them in order.  The memory address is $i$ if the data item is the $i$th item to first appear in the trace.  An AI metric is a trace that standardizes data-to-memory mappings. For example, two traces $abc\ abc$ and $cba\ cba$ have the same AI measure $e_1,e_2,e_3\ e_1,e_2,e_3$.  AI is more abstract than singleton.  If a program is run multiple times with the same input but different memory allocations, the singleton measure changes, but the AI measure does not.

\item \emph{Reuse time (RT) sequence.}  For each access, the reuse time is the increment of logical or physical time since the last access of
  the same datum.  The reuse time is $\infty$ if it is its first access. For a finite reuse time, the minimal is 1 and the maximum
  $n-1$.  The reuse time has been called the inter-reference interval (iri) in the working-set theory~\cite{Denning:CACM68}, inter-reference gap in LIRS~\cite{JiangZ:SIGMetrics02}, and reuse distance in StatCache and StatStack~\cite{Eklov+:HiPEAC11}.
  
\item \emph{Reuse distance (RD) sequence.}  For each access, the reuse distance is the number of distinct data accessed since the last access to the
  same datum.  The reuse distance is $\infty$ if it is its first access.  For a finite reuse distance, the minimal is 1,  because it includes the reused datum, and the maximum is $m$.  The reuse distance is the same as the LRU stack distance~\cite{Mattson+:IBM70}, which is often called stack distance in short.
\end{itemize}

For either RT or RD, the locality may be represented by the entire sequence or be broken down into per-datum sequences:

\def \pd {\emph{pd}}
 
 \begin{itemize}
 \item \emph{Per datum (PD) sequences}, which converts a trace into a set of RT or RD sub-sequences $\pd[e]=(f_e, r_2, \dots, r_{n_e})$ for each datum $e$, where $f_e$ is the time of $e$'s first access, $n_e$ the number of accesses, and $r_i$ the reuse time or reuse distance of $i$th access.  Note that $r_1=\infty$ is omitted, and naturally $\sum_{e \in M} n_{e} = n$.
 \end{itemize}

\subsubsection{Equivalence of Sequence Metrics}
\label{sec:eq}

The equivalence is shown by mutual conversions.  The conversions from AI to other sequences, AI $\rightarrow$ RT, AI $\rightarrow$ RD, RT $\rightarrow$ PD$\cdot$RT and RD $\rightarrow$ PD$\cdot$RD are straightforward, so are PD$\cdot$RT $\rightarrow$ RT and RT $\rightarrow$ AI in the reverse direction from reuse time sequences.
The remaining two conversions are from reuse distance sequences, RD $\rightarrow$ AI and PD$\cdot$RD $\rightarrow$ RD, which are shown by the next two theorems.  

\begin{theorem} \label{thm-rd2ai}
	The address-independent sequence AI can be built from the reuse distance sequence RD.
\end{theorem}

\begin{proof}
The RD sequence is used to drive an LRU stack. When the reuse distance is $\infty$, a new data item $i$ is created and placed on top of the stack (first position).  At a finite reuse distance $x$, the data item at stack position $x$ is moved to the top, and the items in positions $1 \dots x-1$ are moved down by one position.  The AI trace is the sequence of data items that appear at the top position of the stack.
\end{proof}

The construction of an AI trace is more difficult from per datum (PD) reuse distances, because the order of reuses between data items is lost in the PD conversion.

\begin{theorem}\label{thm-rds2trace} 
	The AI trace can be built from per datum reuse distances PD $\cdot$ RD.
\end{theorem}


\begin{proof}
Algorithm~\ref{alg-performance-memorytracerebulid} gives the conversion PD $\cdot$ RD $\rightarrow$ AI, which proves the theorem.
\end{proof}

\begin{algorithm}
\caption{PD$\cdot$RD $\rightarrow$ AI conversion}
\label{alg-performance-memorytracerebulid}
 $lastpos[1\dots m]\leftarrow pd[1\dots m][1]$
 $nextpos[1\dots m]\leftarrow pd[1\dots m][1]$
 $cnt[1\dots m]\leftarrow \{1\}$
\For{$i=1\TO n$}{     
	 $e\leftarrow 0$

	\For{$e'=1\TO m$}{
		\If{$nextpos[e'] = i\,\,\&\&\,\, (e = 0\,\,||\,\,lastpos[e]<lastpos[e'])$}{
			 $e\leftarrow e'$
		}
	}
	\For{$e'=1\TO m$}{
		\If{$lastpos[e']  < lastpos[e]$}{
				 $nextpos[e']\leftarrow nextpos[e']+1$
		}
	}

	 $ai[i] \leftarrow e$
	 $cnt[e]\leftarrow cnt[e]+1$
	 $lastpos[e]\leftarrow  i$
	 $nextpos[e]\leftarrow  i+pd[e][cnt[e]]$
}
\end{algorithm}

The main loop of Algorithm~\ref{alg-performance-memorytracerebulid}, starting at Line
4, constructs the AI trace $\ai[1\dots n]$
by selecting the datum $e$ accessed at each time $i$.  
Lines 1 to 3 initialize the auxiliary data: the last access time $lastpos[e]$ is the time of
$e$'s last access before $i$, $nextpos[e]$ the estimated time of its
next access, the access count $cnt[e]$ the number of times $e$ has
been accessed.  Initially for each datum $e$, the first access is $f_e$, and its
access count $cnt[e]=1$.

The main loop has two inner loops: the selection loop and the update
loop.  The selection loop, Lines 6 to line 10, chooses $e$ for $ai[i]$
if its estimated next access time is $i$.  There may be multiple
choices.  The selection loop does not stop at the first such datum.  It finds every such item and chooses the one with the largest last access time.  This is a choice based on recency, i.e. most recent last access.  Naturally, this choice is unique. 

The update loop is the second inner loop. Lines 11 to 15 update $nextpos$
for all other elements $e'$. If $e$ has been accessed after the last $e'$, the $e$ access is a recurrence,
so the estimated next access time of $e'$ is increased by 1.  
Then, Lines 16 to 19 update for $e$: the current access is now the last
access, the access count $cnt[e]$ is increased by 1, and the next access time is
estimated to be the current time plus the next reuse distance $\pd[e][cnt[e]]$.  
\medskip

The PD$\cdot$RD $\rightarrow$ AI conversion has two requirements in addition to per datum reuse distances.  First, the recency choice is necessary.  Consider the AI trace $(e_1,e_2,e_3,e_2,e_1)$.  When time $i=4$, the next access times of $e_1,e_2$ are both estimated as 4.  The selection loop must choose $e_2$, which is more recently accessed.  Second, the first access time is necessary. Consider two AI traces $e_1,e_1,e_2,e_3,e_2$ and $e_1,e_2,e_1,e_3,e_3$ that have the same per-datum reuse distances.  Without the first-access times, no algorithm can distinguish between them.





\subsubsection{Histogram Metrics}
\label{sec:hist}


The histogram construction (HI) produces two types of histograms:

\begin{itemize}
\item The RT histogram $\textit{rt}(i)$, which counts the number of reuse times that equal to $i$, $i = 1, \dots, n-1, \infty$ and $0 \le \textit{rt}(i) \le n$.

\item The RD histogram $\textit{rd}(i)$, which counts the number of reuse distances that equal to $i$, $i = 1, \dots, m, \infty$ and $0 \le \textit{rd}(i) \le n$.
\end{itemize}

HI conversion loses all information about memory address, access time, and order of reuses.  Instead, it sorts reuses by their time or distance.  A histogram can be interpreted as a probability distribution.  

\begin{itemize}
\item The probability function $P(x \le y) = \frac{\Sigma_{i=1}^y x(i)}{n}$, where $x$ may be $\textit{rt}$ or $\textit{rd}$, and $0
  \le P(x \le y) \le 1$.
\end{itemize}

The reuse time histogram was called the interreference density, and its probability function the interreference distribution~\cite{DenningS:CACM72}.  The reuse distance histogram was called the locality signature~\cite{Zhong+:TOPLAS09}.  \citet{Gupta+:JPDC13} used both to define locality as the probability of reuse, where the two types of histograms give the likelihood of reuse in next-n-addresses and next-n-unique-address.

We show two invariances of the reuse time histogram.  

	\begin{lemma} \label{lemma-1}$\sum\limits_{i=1}^{n-1}\textit{rt}(i) = n-m$
	 \end{lemma}
	 
	\begin{lemma} \label{lemma-2}$\sum\limits_{i=1}^{n-1}i\times \textit{rt}(i) = \sum\limits_{e=1}^m (\lte-\fte)$
	 \end{lemma}

From the proof of the Theorem \ref{thm-rds2trace},
if an access of element $e$  reuse distance $i$,
and all the $m-i$ elements except the $i$ distinct elements
have already shown in the trace before and will be accessed again after,
their $nextpos$ will be increased by one,
which means the next access of the $m-i$ elements
will have an repetitive access of $e$.
We define \textbf{repetitive access} $rep$ as:
\[rep(i) = \textit{rd}(m-i)\]

The repetitive access histogram is the reversal of reuse distance histogram.

We define the \textbf{sealed memory trace}, if the first $m$ and
last $m$ elements are permutations of all $m$ elements of the trace's
data set.  Namely, for every element $e$, $\fte\le m$ and
$\lte\ge n-m+1$.  In a sealed memory trace of $n$ accesses to
$m$ data, there are $n-m$ reuses.  The following theorem shows a
property of the total reuse time, $\sum_{i=1}^{n-1}i\times \textit{rt}(i)$.  It
shows that in a sealed memory trace, the average reuse time is $m$.  The proof is
based on the observation that the reuse time, $m$, comes from either
reuse distance or repetitive access.

	\begin{theorem} \label{thm-normal-rt}
	The average reuse time of a sealed memory trace is $m$.
	\[\sum_{i=1}^{n-1}i\times \textit{rt}(i) = m (n-m)\]
	 \end{theorem}

\begin{proof}
It is obvious that if the $rep$ of an access in $ai(m+1\dots n-m)$ is $i$,
it contributes $i$ repetitive accesses to the total $rt$.
This is true because all the $m$ elements are shown before and after this access.
However, this is not true for the last $m$ accesses, since not all elements will be accessed again.
Let's first examine a particular sealed memory trace, where
the reuse distances of all elements' last accesses are $m$ and their $rep$ is 0 consequentially.

Thus, for this particular sealed memory trace, besides $rd$s, all $rep(i)$ contributes $i\times rep(i)$ to $rt$. We have:

{\setlength\arraycolsep{2pt}
\begin{eqnarray}
\sum_{i=1}^{n-1}i\times \textit{rt}(i)  &=& \sum_{i=1}^{m}i\times \textit{rep}(i) + \sum_{i=1}^{m} i\times \textit{rd}(i) \nonumber\\
&=& \sum_{i=1}^{m}i\times \textit{rd}(m-i) + \sum_{i=1}^{m} i\times \textit{rd}(i) \nonumber\\
 &=& \sum_{i=1}^{m}(m-i)\times \textit{rd}(i) + \sum_{i=1}^{m} i\times \textit{rd}(i) \nonumber\\
 &=& m \sum_{i=1}^{m}\textit{rd}(i) = m(n-m)\nonumber
\end{eqnarray}}

Then, we can show that if the last $m$ elements are permuted,
it doesn't change the average reuse time.
There is an easy way to prove this.
Let's make a new memory trace $ai'(1\dots n+m)$,
whose first $n$ accesses $ai'(1\dots n)$ 
is $\textit{ai}(1\dots n)$.
The tail side of $ai'$ duplicate the corresponding
tail side of $ai$. Thus the new trace 
satisfies the requirements mentioned above and
its total reuse time is $m(n+m-m)$.
Now let's remove tail side of $ai'$.
Since $ai'(n+1\dots n+m)=ai'(n-m+1\dots n)$,
the total reuse time of $ai'(n+1\dots n+m)$,
which should be decreased, is $m^2$.
The total reuse time of the original trace is therefore $m(n+m-m)-m^2=m(n-m)$.
\end{proof}

\subsubsection{Types of Histograms}
\label{sec:hist-types}


A histogram is more space efficient than a sequence.  The histogram
construction can be viewed as having two steps: sorting the accesses
by their locality value, e.g. reuse time or reuse distance, and then
counting the number of accesses with the same value.  For a greater
saving, a third step is to group consecutive locality values into a
single bin.  Instead of one counter for each locality value, one
counter is used for a range of values.  The locality range is an
approximation but it bounds the maximal error.  If we grow the range
exponentially, we reduce the size of the histogram logarithmically
while ensuring a relative precision.  An example is the log-linear histogram, where the range of values grows by powers of two: 1, 2,
3--4, 5--8, etc.  The large ranges, e.g. 1025--2048, are evenly
divided into a fixed number of smaller ranges,
e.g. 256~\cite{Xiang+:PACT11}.  The asymptotic space cost is
logarithmic, and the approximation is equivalent to recording the most
significant digits of locality values.

Histogram locality can be stored in constant space. \citet{Zhong+:TOPLAS09} sorted program accesses by locality and then divided them into
equal-size groups, for example, 1000 groups each containing 0.1\% of
memory accesses.  This grouping limited the
effect of error from any single group.  The imprecision came from
the spread of locality values within a group.  \citet{MarinM:SIGMETRICS04} controlled the locality spread by recursively dividing
a group until its range of values was within a
limit.  \citet{Fang+:PACT05} improved the precision
in coarse-grained histograms, i.e. large spreads, by approximating it
with a distribution.  They showed that the linear
distribution was a more effective approximation than the uniform
distribution.

Table~\ref{tbl:hist} compares the space requirement of locality
sequences and histograms.  While the sequence locality takes linear
space and cannot be approximated, the histogram locality takes either
logarithmic or constant space, with controlled loss of information as
discussed in this section.


\begin{table}[h]
\tbl{Space requirements of sequence and histogram locality.\label{tbl:hist}}{
\centering
\begin{tabular}{|c|c||c||c|c|} 
\hline
\multicolumn{2}{|c|}{} & RT/RD sequences & RT histogram $\textit{rt}(w)$ & RD histogram $\textit{rd}(v)$\\ \hline \hline
\multicolumn{2}{|c|}{indexing parameter} & time $t \in [1\dots n]$  & interval $w \in [0\dots n]$ & volume $v \in [1\dots m]$ \\ \hline
space & accurate & $O(n)$ & $O(n)$  & $O(m)$\\ \cline{2-5}
complexity & approx. & $O(n)$ & $O(\log n)$, $O(1)$ & $O(\log m)$, $O(1)$ \\ \hline
\end{tabular}}
\end{table}

\subsubsection{Strength and Limitations}



By definition, locality is essentially a pattern of reuse. The metrics in this section represent data reuse with different levels of abstraction.  Singleton traces use exact memory addresses.  AI traces use abstract memory addresses.  Reuse distance and reuse time dispense with the address of reuses but still retain their order.  The reuse distance histogram is the most abstract and compact because it removes all information about the memory address, the access time and the order of reuses.  
This high level of abstraction has both strengths and limitations. 

In many important problems, the reuse distance histogram is an adequate and the most compact measure of locality.  In cache analysis, it gives the miss ratio of the fully associative cache~\cite{Mattson+:IBM70}, direct-mapped or set-associative cache~\cite{Smith:ICSE76,Nugteren+:HPCA14}, and cache with other reuse-based replacement policies~\cite{SenW:SIGMETRICS13} of all sizes.  It is used to separate the locality effect by the program structure~\cite{MarinM:SIGMETRICS04} and the load/store operation~\cite{Fang+:PACT05}, model the change of locality as a function of the input~\cite{Zhong+:TOPLAS09,MarinM:SIGMETRICS04,Fang+:PACT05} and the degree of parallelism~\cite{WuY:PACT11}, and predict the performance of different cache designs and parameters~\cite{Zhong+:TOC07,Wu+:ISCA13}, making it the most widely used metric of access locality.

However, there are two limitations. The first problem occurs when analyzing program interaction in shared cache. \citet{Xiang+:PACT11,Ding+:JCST14} gave an example showing that when two program traces are interleaved into a single trace and the exact interleaving is known, e.g. uniform interleaving, we could not infer the reuse distance histogram of the interleaved trace from the reuse distance histograms of the individual traces. In other words, we cannot compute the combined locality from those of the components.  On modern multicore processors where cache is increasingly shared, this lack of composability is a serious limitation.


Interestingly, all other access metrics are composable. When the method of interleaving is given, an interleaved AI trace can be easily constructed from individual AI traces.  From the equivalence theorems, all other sequence metrics, the reuse distance and reuse time sequences and their per datum sequences, are composable.  Moreover, the reuse time histogram is composable: the reuse time histogram of an interleaved trace is the sum of the reuse time histograms of the individual traces, if all reuse times are normalized to include the effect of interleaving.  

The second limitation of histograms is the loss of information about phase behavior. \citet{BatsonM:SIGMETRICS76} and \citet{Shen+:JPDC07} used reuse distances to capture and characterize program phases.  While the loss of phase information is in both types of histograms, the loss of composability is only for the reuse distance histogram. The difference in composability is another demonstration of the non-equivalence between the reuse time and reuse distance histograms. It also shows that the second limitation is not the cause of the first.  




Next we introduce a set of locality metrics which are both composable and compact.  We will use the reuse time histogram not as a metric of locality but the basis to derive other locality metrics.

\subsection{Timescale Metrics}
\label{sec:timescale}

A timescale is a length of time, which may be measured in seconds or years in physical time or number of memory accesses in logical time.  A timescale metric is a mathematical function $f(x)$, where $x$ ranges across all timescales, i.e. $x \ge 0$.



\subsubsection{Footprint}

\def \reuse {\textit{reuse}}

In an execution, every consecutive sub-sequence of accesses is a time window, formally as $(t,x)$, where $t$ is the end position and $x$ the window length.  The number of distinct elements in the window is the \emph{working-set size} $\omega(i,x)$~\cite{Denning:CACM68}.  For a length $x$, the footprint $\fp(x)$ is the average working-set size, computed by the total working-set size divided by the number of length-$x$ windows:

\begin{align}
\label{eq:fp-def}
\fp(x) = \frac{1}{n-x+1}\sum_{t=x}^{n} \omega(t,x)
\end{align}

\noindent The footprint measures the average working-set size in all timescales and shows the growth of program working set over time.

\subsubsection{Computing the Footprint}
\label{sec:rt2fp}

Xiaoya Xiang gave the following formula to compute the footprint from reuse times and the times of first and last accesses~\citep{Xiang+:PACT11}.  

\predisplaypenalty=10000
\postdisplaypenalty=1549
\displaywidowpenalty=1602

\begin{align}
\label{eq:xiang}
\fp(x)& = m - \frac{1}{n-x+1} \bigg( \sum_{i=x+1}^{n-1}(i - x) \rt(i) \nonumber \\
 &+\sum_{k=1}^{m}(f_k - x)I(f_k>x) \nonumber \\
 &+\sum_{k=1}^{m}(l_k-x)I(l_k>x) \bigg)
\end{align}

\noindent The symbols in the Xiang formula are:
\begin{itemize}
\item $\rt(i)$: the number of accesses whose reuse time is $i$.
\item $f_k$: the first access time of the $k$-th datum (counting from 1).
\item $l_k$: the \emph{reverse} last access time of the $k$-th datum.
  If the last access is at position $x$, $l_k = n+1-x$, that is, the
  first access time in the reverse trace (counting from 1).
\item $I(p)$: the predicate function equals to 1 if $p$ is true; otherwise 0.
\end{itemize}

\citet{Xiang+:PACT11} used two pages in their paper to derive the formula based on ``differential counting" of how the working set changes over successive windows.  Next is a new, shorter explanation.  The idea is ``absence counting'', by starting with assumption of all data in all windows and then counting all absences and subtracting their effects.  For people who have filed income tax in the United States, taking deductions is a familiar process.


The first deduction is based on data reuses.  If a reuse time $i$ is greater than $x$, there are $i-x$ windows of length $x$ that do not access the reused datum.  The working-set size should be reduced by $i-x$ to account for this absence.  The total absence from all reuses is $\sum_{i=x+1}^{n-1}(i-x) rt(i)$.

The next two deductions follow a similar rationale.  If the $k$th datum is first accessed at time $f_k$ and $f_k > x$, it is absent in the first $f_k - x$ windows of length $x$.  Similarly, if it is last accessed at $l_k$ counting backwards and $l_k>x$, it is absent in the last $l_k-x$ windows of length $x$.  The total adjustment are shown by the last two terms of the Xiang formula.



\subsubsection{The Denning-Schwartz Formula}
\label{sec:ds}


The first timescale metric of locality is the average working-set size (WSS) $s(x)$ formulated by \citet{DenningS:CACM72}.\footnote{Although both define average WSS, mathematically footprint in Eq.~\ref{eq:fp-def} differs from Denning and Schwartz in Eq.~\ref{eq:denning-schwartz}.}  The Denning-Schwartz formula is inductive: the WSS at $x$ is the WSS at $x-1$ plus the miss ratio.  In the base case, we have an empty working set $s(0) = 0$ and 100\% miss ratio $m(0)=1$.  At window length $x$, an access is a miss if its reuse time $t$ is greater than $x$, that is, $m(x) = P(t>x)$.  This type of miss ratio is called the \emph{time-window miss ratio}.

\begin{align}
\label{eq:denning-schwartz}
s(x) = s(x-1) + m(x-1) = \sum_{i=0}^{x-1} m(i) = \sum_{i=0}^{x-1} P(\rt>i) 
\end{align}


\sout{Unlike footprint, in particular Eq.~\ref{eq:fp-def}, Eq.~\ref{eq:denning-schwartz} is not directly related to WSS or reuses in individual windows.  }The mathematical constructions of footprint and average working set differ: Denning-Schwartz formula is additive, while Xiang is subtractive. Next, we show an underlying equivalence.

\subsubsection{Steady-state Footprint}
\label{sec:timescale-props}

\def \ssfp {\textit{ss-fp}}

Steady-state footprint is the average WSS not considering the effect of trace boundaries.  In the Xiang formula, the first and last access times affect two of the terms.  If we drop these two terms and use $n$ as the window count, we say that the revised formula computes the \emph{steady-state footprint}  $\ssfp(x)$:

\begin{align}
\label{eq:ssfp}
\ssfp(x) = m - \frac{\sum_{i=x+1}^{n-1}(i - x) \rt(i)}{n}  = m - \sum_{i=x+1}^{n-1}(i - x) P(\rt = i)
\end{align}

If a trace is infinitely long $n=\infty$, the footprint is $\lim_{n \to \infty} \fp(x)$.  It is easy to see that the limit footprint is the steady-state footprint when $n \to \infty$.

\begin{align}
\label{eq:fp-inf}
\lim_{n \to \infty} \fp(x) & = \ssfp(x) = m - \sum_{i=x+1}^{\infty}(i - x) P(\rt = i)
\end{align}

\begin{theorem}
(\emph{WS-FP Equivalence})
\label{thm-ws-fp-eq}
The Denning-Schwartz formula computes the steady-state footprint, i.e. $s(x) = \ssfp(x)$ for all $x \ge 0$.
\end{theorem}

\begin{proof}
The equivalence is proved by induction.  In the base case, $s(0)= \ssfp(0) = 0$.  Assuming $s(x) = \ssfp(x)$, we see they increase by the same amount 

\begin{align*}
s(x+1) - s(x) = & \sum_{i=0}^{x+1} P(\rt>i) - \sum_{i=0}^{x} P(\rt>i) = P(\rt > x)
\end{align*}

\begin{align*}
\ssfp(x+1) - \ssfp(x) = & \ m - \sum_{i=x+2}^{n}(i - x - 1) P(\rt = i) \\
  & - (m - \sum_{i=x+1}^{n}(i - x) P(\rt = i) ) \\
  = & \ P( \rt > x) \\
\end{align*}

\noindent Hence, $s(x+1) = \ssfp(x+1)$, and the equivalence holds for all $x \ge 0$.
\end{proof}

Consider an example trace $abc\ abc\ \dots$  We have $P(\rt = i) = 1$ for $i=3$ and $P(\rt = i) = 0$ otherwise.  The steady-state footprint, computed by either Eq.~\ref{eq:denning-schwartz} or Eq.~\ref{eq:ssfp}, is $\ssfp(w)=0, 1, 2$ for $x=0, 1, 2$ and 3 for $x \ge 3$.

Because of the equivalence, we can easily prove the following:

\begin{theorem}
\label{thm:ssfp-concavity}
$\ssfp(x)$ is bounded and concave.
\end{theorem}

\begin{proof}
Eq.~\ref{eq:ssfp} shows $\ssfp(x) \le m$, so it is bounded.  Eq.~\ref{eq:denning-schwartz} shows $\ssfp(x) \le m$ increases by $P(\rt>x)$ at each $x$. Since its derivative is monotonically decreasing with $x$, $\ssfp(x)$ is concave.
\end{proof}

\noindent This concavity implies strict monotonicity until it reaches the maximum value, which is a common shape in the steady-state footprint of \emph{all} programs:

\begin{corollary}
\label{cor:fp-shape}
$\ssfp(x)$ starts from 0, is strictly monotone until it increases to $m$, and then stays constant.
\end{corollary}

The concavity has a critical importance later in the section on cache metrics.  It will ensure that miss ratios are monotone, i.e. no Belady anomaly~\cite{Belady66}, and a cache metric, cache fill time, exists and is unique.


\subsubsection{Observational Stochastics}

\citet{DenningB:CSUR78} formulated a new theory of operational analysis called \emph{observational stochastics}.  Conventional analysis was based on classic queuing models with idealistic assumptions such as infinite stationary processes.  Observational stochastics are based on directly measurable variables and directly verifiable assumptions.  The theory and applications in system and network analysis are enunciated in two recent books~\citep{Buzen:Book15,DenningM:Book15}.  All locality metrics and properties in this paper are based on direct measurements, do not depend on idealistic assumptions, hence are extensions of observational stochastics. 


The original timescale metric, Denning-Schwartz, was derived based on stochastic assumptions --- that a trace is infinite and generated by a stationary Markov process, i.e. a limit value exists~\citep{DenningS:CACM72}. In later work Denning and his colleagues adopted observational stochastics and used the formula on finite-length traces, with adjustments to account for boundary effects~\citep{SlutzT:CACM74,DenningS:CACM78}.  Although the formula was found accurate, this accuracy is not justified by the original derivation, because the stochastic assumptions cannot be proved for real programs. 

The properties of steady-state footprint, which is the same as Denning-Schwartz, give new theoretical explanations to this accuracy.  First, Theorem~\ref{thm-ws-fp-eq} shows that Denning-Schwartz, without any adjustment, accurately computes the steady-state footprint for any execution trace, whether finite or infinite.  Second, Eq.~\ref{eq:fp-inf} and Theorem~\ref{thm-ws-fp-eq} show that Denning-Schwartz is the footprint of infinite-long traces, even when the limit does not exists or is not unique.  Consider a sequence of accesses of two elements divided into pieces separated by commas: 1, 2, 22, 1212, 2222222, $\dots$, where half of the pieces alternate between 1 and 2, half of the pieces are all 2, and the length of every piece is the length of the entire trace before.  The footprint of this trace has two limit values, which Denning-Schwartz can compute even though this violates the stochastic assumption from which it was originally derived.



In addition, steady-state footprint expands the theoretical results of footprint in two ways.  First, Theorem~\ref{thm-ws-fp-eq} shows the precise relation between the two timescale metrics: Denning-Schwartz is an overestimate of footprint, and the difference is given by the two terms in the Xiang formula.  Second, the equivalence theorem leads to a different and simpler proof of concavity than \citet{Xiang+:ASPLOS13}

\subsubsection{From Footprint to Reuse Time}

\def \w {\mathcal{W}}

Denote total working-set size as $\w(x) =(n-x+1)\fp(x)$.  Using the Xiang formula, the first and second order finite differences of $\w(x)$ are:
{\setlength\arraycolsep{2pt}
\begin{eqnarray}
\nabla \w(x+1) &=& \w(x+1) - \w(x) = m+ \sum_{i=x+1}^n\textit{rt}(i)-\sum_{\fte<x+1}1 -\sum_{n-x<\lte}1 \nonumber\\
\nabla^2 \w(x+1)&=& \nabla \w(x+1) -\nabla \w(x)\nonumber\\
 &=&  -\rt(x)-\sum_{e}I(\fte=x)-\sum_{e}I(\lte=n-x+1) \nonumber
\end{eqnarray}}

Therefore, footprint can be used to derive the reuse time histogram if the first and last access times are known.

\subsection{Cache Metrics}

\def \res {\textit{res}}

The following metrics are average quantities of events in fully-associative LRU cache of size $c$:

\begin{itemize}
\item \emph{miss ratio} $\mr(c)$, which is the average rate of cache misses.  
\item \emph{inter-miss time} $\im(c) = \frac{1}{\mr(c)}$, which is the average time between two consecutive misses.
\item \emph{fill time} $\ft(c)$, which is the average time for the first $c$ misses to happen in an empty (cold-start) cache.
\item \emph{residence time} $\res(c)$, which is the average time a data item stays in the cache.
\end{itemize}

An analysis may consider all nonnegative integer cache sizes for two reasons.  In practice, cache is often shared, and the occupancy of a program in shared cache can be any size depending on its peers.  Fully analysis must measure the effect of locality at the granularity of a single cache block.  In modeling, the miss ratio curve for fully associative cache is equivalent to reuse distance~\cite{Mattson+:IBM70,Xiang+:ASPLOS13}, which can model the effect of cache associativity~\cite{Smith:ICSE76} and non-LRU replacement~\cite{SenW:SIGMETRICS13}.

\subsubsection{Footprint Differentiation}

Footprint differentiation computes the miss ratio as the derivative of the footprint (or the steady-state footprint).  

\begin{align}
\label{eq:fp-diff}
\mr(\fp(x)) = \fp'(x) = \fp(x+1) - \fp(x) \nonumber
\end{align}

\noindent The equation is the same as the time-window miss ratio formula of \citet{DenningS:CACM72} except by replacing the window length $x$ with its footprint $\fp(x)$.   

When steady-state footprint $\ssfp(x)$ is used, the derivative is monotone. The monotonicity is critically important for two reasons.  First, it is a prerequisite for correctness since real LRU cache has monotone miss ratios, i.e. no Belady anomaly.  Second, the miss ratio function $\mr(\ssfp(x))$ is discrete and not continuous.  It is not defined on all cache sizes.  In fact, $\ssfp(x)$ is usually not an integer, but an actual cache size must be.  The monotonicity bounds the miss ratios for missing cache sizes.  

The following theorem shows that footprint differentiation computes the miss ratio for all actual cache sizes $c$.

\begin{theorem}
\emph{(Footprint Differentiation)}
\label{thm:fp-mr}
\begin{align}
\ssfp'(x) \le \mr(c) < \ssfp'(x+1) \text{ if } x \le c < x+1
\end{align}
\end{theorem}

\noindent The proof follows directly from the monotonicity of $\mr(c)$, which follows directly from Theorem~\ref{thm:ssfp-concavity}.  The steady-state footprint is often a fractional number.  Theorem~\ref{thm:fp-mr} shows that its derivatives at $x$ and $x+1$ are ``poles" that mark the bounds of the miss ratio of cache sizes $c$ between $x$ and $x+1$, integer or not.  In practice, the miss ratio is selected by the $x$ whose footprint is closest to $c$.

As stated in Corollary~\ref{cor:fp-shape}, the steady-state footprint of all programs have a common shape, which starts from 0 and increases continuously with diminishing increments until it reaches $m$.  Its derivative, the miss ratio, starts from 100\% when $c=0$ and decreases monotonically until it drops to 0\% when $c=m$.  The lower and upper bounds are guaranteed and ensure any miss ratio it computes is valid. 

\medskip

\def \ws {\mathcal{W}}

To understand footprint differentiation, consider a memory access and the factors that cause it to be a cache hit or miss.  Instead of reuse distance in access metrics, consider the execution window $w$ preceding the access, such that its working set size equals to the cache size, i.e. $|\ws|=c$.  At the access, the cache is full and filled with (only) the data of $\ws$.  The access is a cache miss if and only if the accessed data is outside $\ws$.  The number of misses is the number of $w$ windows followed by such an access.  At a miss, $\ws$ grows after the access.  The miss ratio is in fact the average growth of the working-set size.

Footprint is the average working-set size.  Its derivative is the growth of the average working-set size.  The essence of footprint differentiation is to \emph{equate the average growth of the working-set size with the growth of the average working-set size}.  In other words, the miss ratio equals to the growth of footprint.

Because of the equality, we can use the miss ratio to construct the footprint.  The following shows that after one access, the footprint increases by 1 if the access is a miss and 0 otherwise. 


\begin{align*}
\fp(w+1) = \mr(c) (\fp(w) +1) + (1- \mr(c)) \fp(w)
\end{align*}

\noindent The equation is mathematically identical to footprint differentiation.

As an example, consider the access trace \texttt{abc abc abc}.  Table~\ref{tbl:fp-diff-example} shows the footprint in the second row and the miss ratio in the third row, computed as the difference between consecutive values in the second row.

\begin{table}[h]
\tbl{The steady-state footprint of \texttt{abc abc abc} and the miss ratio computed using footprint differentiation.
\label{tbl:fp-diff-example}}{  
\centering
\begin{tabular}{|c||c|c|c|c|c|}
  \hline
  $x$ & 0 & 1 & 2 & 3 & $\ge 4$ \\ \hline
  $c = \ssfp(x)$ & 0 & 1 & 2 & 3 & $\ge 4$ \\ \hline 
  $\mr(c)$ & 100\% & 100\% & 100\% & 0\% & 0\% \\ \hline
\end{tabular}}
\end{table}

\subsubsection{Cache Fill Time}
\label{sec:ft}

In the higher-order theory of locality (HOTL), \citet{Xiang+:ASPLOS13} defined \emph{cache fill time}, which we denote as $\ft(c)$, as the average amount of time for a program to access an amount of data equal to a cache size $c$.  Xiang et al. studied two definitions and chose to define it as the inverse function of footprint $\fp(x)$.  The inverse is unique because of concavity (excluding the fill time at or greater than $m$).


\citet{Hu+:USENIX16} defined the \emph{average eviction time} (AET) as the average time between the last access of a data block in cache and its eviction from the cache.  Trivially, fill time $\ft(c)$ is the average eviction time of fully-associative LRU cache of size $c$.  Indeed, Hu et al. showed that ignoring boundary effects, Denning-Schwartz (Section~\ref{sec:ds}) computes the average eviction time (AET).  

Figure~\ref{fig:fpft} shows that fill time and footprint are the two-way mapping between the dimensions of space and time, i.e. between cache size and window length.  Footprint differentiation computes the miss ratio using the space dimension, i.e. $\mr(c)$ is the fractional value of $\fp(x+1)-\fp(x)$. \citet{Xiang+:ASPLOS13} gives another method, \emph{reuse-time conversion}, which computes the miss ratio using the time dimension.

\begin{figure}[t]
\centering
\includegraphics[width=10cm]{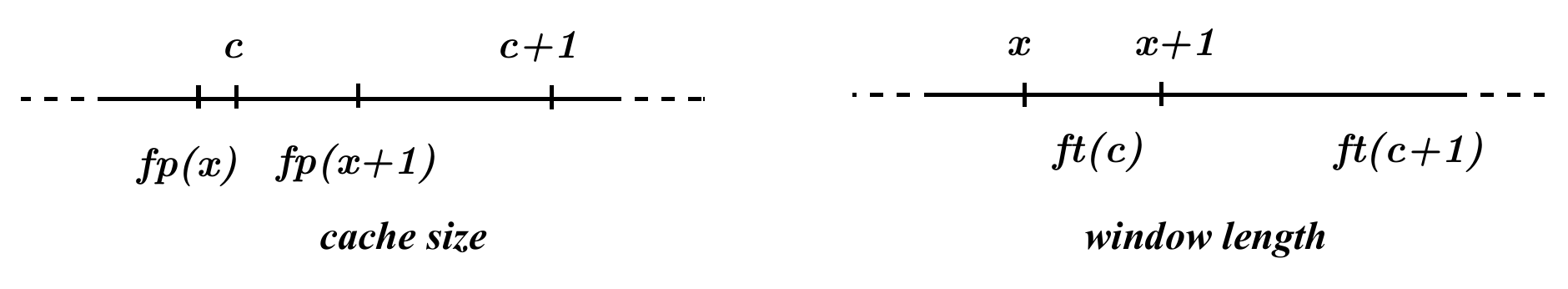}
\caption{Footprint $\fp(x)$ and fill time $\ft(c)$ are inverse functions: $x = \ft(\fp(x))$ and $c = \fp(\ft(c))$ for all $x,c \ge 0$.}
\label{fig:fpft}
\end{figure}

Given cache size $c$ and its fill time $\ft(c)$, an approximation can be made such that an access is a miss if and only if its reuse time is greater than $\ft(c)$.  The miss ratio is:

\begin{align}
\label{eq:rtc}
\mr(c) = P( rt > \ft(c) )
\end{align}

In theory, \citet{Xiang+:ASPLOS13} showed that reuse-time conversion computes the same result as footprint differentiation (Eq.~\ref{eq:fp-diff}), when the boundary effect is negligible, i.e. for the steady-state footprint (Eq.~\ref{eq:ssfp}).

In practice, reuse-time conversion has two benefits.  First, it counts the cold-start misses correctly.  These are first accesses whose reuse time is infinite, since $\rt > \ft(c)$ for all $c$.  Second, in short traces, e.g. sampled segments, the boundary effect is significant.  When it is included, differentiation of the footprint is not guaranteed monotone.  Reuse-time conversion, however, guarantees monotone miss ratios.  

Sampling makes online analysis possible.  For SPEC CPU benchmarks, footprint sampling has been shown to reduce overhead to less than 1\% visible~\cite{Xiang+:ASPLOS13} and less than 0.09 seconds per program on average~\cite{Wang+:CCGrid15} and improve accuracy for programs with phase behavior~\cite{Wang+:CCGrid15}.

A third, theoretical benefit of reuse-time conversion is the justification of an early method by Easton and Fagin.


\subsubsection{Easton-Fagin Recipe}
\label{sec:recipe}
\label{sec:snir}

\citet{EastonF:CACM78} were among the first to study cache sharing, in particular, the effect of context switching in cache. They defined a ``cold-start cache" as one when a program is switched back and its earlier data have been wiped out, and to distinguish from it, a ``warm-start cache" for a regular, solo-use cache.

The 1978 paper gave an ingenious solution to a practical problem: to compute
the cold-start miss ratio, which is difficult to simulate, from the miss ratio of warm-start cache, for which simulation is easy.  It was ground breaking and pioneered the approach to compute the shared cache performance by reusing the existing solutions already developed for non-shared cache.

Easton and Fagin ``gave a rough explanation as to why our recipe is reasonable" but ``remarks without proof that this need not be the case, even in the LRU stack model.''  However, they found that the ``estimate was almost always
within 10-15 percent of the directly observed average cold-start miss
ratio.'' Next we derive the recipe from the measurement theory.

Cache fill time $\ft(c)$ is the time it takes a program to have the first $c$ misses in cold-start cache of an infinite size.\footnote{The lifetime of first $c$ misses in cold-start cache of size $c$, $\textit{LIFE*}_c(c)$, in \citeauthor{EastonF:CACM78} is fill time $\ft(c)$ in this paper, and the lifetime (of 1 miss) in warm-start cache of size $c$, $\textit{LIFE}_c(1)$, is $\im(c)$.}  The Easton-Fagin recipe can be written as follows:

\begin{align}\label{eq-ef}
\ft(c) \approx \sum_{i=0}^{c-1} \im(i) 
\end{align}

The recipe states that in a cache of size $c$, the fill time is the sum of the inter-miss time in the cache of \emph{all} smaller sizes.  The following formula explains the recipe using the measurement theory

\begin{align*}
\ft(c) = \sum_{i=0}^{c-1} (\ft(i+1) - \ft(i)) \approx \sum_{i=0}^{c-1} \im(i)
\end{align*}

The preceding formula first rewrites $\ft(c)$ into a series of sums and then replaces $\mr(i)$ with the derivative of footprints at the time window of length $\ft(i)$.  The approximation by Easton and Fagin is reduced to the following simpler form: 

\begin{align*}
\ft(i+1) - \ft(i) \approx \im(i) = \frac{1}{\mr(i)} = \frac{\fp(\ft(i+1)) - \fp(\ft(i))}{\fp(\ft(i)+1) - \fp(\ft(i))}
\end{align*}

The first two terms show that the derivative of fill time at cache size $i$ is approximated as the inter-miss time at $i$.  This approximation is explained in the last term, which is a ratio.  The only difference is $\fp(\ft(i+1))$ in the enumerator and $\fp(\ft(i)+1)$ in the denominator.  Both are increases in footprint from the same starting point when the footprint is the cache size  $\fp(\ft(i))=i$.  The denominator is the increase of footprint at $i$ by a unit time, and the enumerator is the increase to fill the cache size $\fp(\ft(i+1))=i+1$.  This is linear approximation at each $i$ --- the rate of footprint increase is constant between $i$ and $i+1$.

Therefore, this section has shown that \emph{the Easton-Fagin recipe is a piecewise linear approximation of the footprint}.  This also explains the flexibility of Easton-Fagin.  The granularity of piecewise approximation can be fine, e.g. consecutive cache sizes, or coarse, e.g. power-of-two sizes.






\subsubsection{Residence Time}

We define the \emph{residence time}, $\textit{res}(c)$, as the average
time a data block stays in cache, from the time of loading into the
cache to the time of eviction. 
The residence time can be computed as follows.  Assuming a fully occupied cache of size $c$ for a time period $T$, the sum of residence time of all data blocks is $Tc$, and the number of data blocks loaded in the cache is $T\textit{mr}(c)$.  The average residence time is $\textit{res}(c) = \frac{T c} {T \textit{mr}(c) } = c / \textit{mr}(c)$.  The same formula can be derived from the Little's law $L=\lambda W$, taking the residence time as the service time $W$, the miss ratio as the arrival rate $\lambda$, and the cache size as the average number of customers in a stable system $L$~\cite[pp. 182]{DenningM:Book15}.

\subsection{Linear-time MRC Modeling}
\label{sec:uni}

A weakness of the Xiang formula (Equation \ref{eq:xiang}) is that the entire reuse-time histogram is required when computing the footprint of any timescale $x$.  The total time to compute the complete footprint $\fp(x)$ for all timescales is quadratic.\footnote{The Denning-Schwartz formula can compute the steady-state footprint in linear time but not the footprint.}  This section derives an additive formula of $\fp(x)$ and then an incremental version of the additive formula that computes the complete footprint $\fp(x)$ in linear time.

\subsubsection{The Additive Formula}
\label{sec:add}


To derive an additive formula, we calculate the footprint 
based on the following observation:
if an element $e$ appears more than once in a window,
we count only its first appearance in its working-set size.

For a datum $e$, let $\textit{ai}(\fte)$ be the initial access in the trace. There are three cases:

\begin{enumerate}
\item If $\fte<w$, $\textit{ai}(\fte)$ is the first appearance of $e$ in the first $\fte$ windows of length $w$.
\item If $w\le \fte\le n-w+1$, $\textit{ai}(\fte)$  is the first appearance in $w$ windows of length $w$.
\item If $ n-w+1 < \fte$, $\textit{ai}(\fte)$ appears first in the last $n-\fte+1$ windows of length $w$.
\end{enumerate}

\noindent Adding this count for all $e$, we have the total footprint contribution from initial accesses, which we denote as $\fp_1(w)$:

\[fp_1(w) = \sum_{\fte\le n-w+1} \min(\fte,w) + \sum_{\fte> n-w+1} (n- \fte+1)\]


Next we consider $j$th access ($2\le j\le k$) of $e$, i.e., $\textit{ai}(t_e^j)$.  There are four cases.

\begin{enumerate}
\item If $t_e^j\le w$, it is the first appearance in $rt_e^j$ ($rt_e^j<w$) windows of length $w$. 
\item If $w\le t_e^j\le n-w+1$, it is the first appearance in $\min(rt_e^j,w)$ windows of length $w$.
\item If $t_e^{j-1}< n-w+1 < t_e^j$, it is the first appearance in $\min(rt_e^j,w)-(t_e^j-(n-w+1))$ windows.
\item If $ n-w+1 \le t_e^{j-1} < t_e^j$, it is not the first appearance in any window of length $w$,
because $\textit{ai}(t_e^{j-1})$ and $\textit{ai}(t_e^{j})$ always appear together in such windows (the last $n- t_e^j+1$ windows of length $w$). We express this zero as $rt_e^j-rt_e^j$.
\end{enumerate}

We merge all the above terms except the subtractive term in the third and fourth cases.  We call the sum $\fp_2(w)$, which can be written concisely as

\[\fp_2(w) = \sum_{i=1}^{n}\min(i,w) \times\textit{rt}(i) \]

We calculate the negative term in the third and fourth case as $\fp_3(w)$.  Note that the
problem is now isolated since the access $t_e^j$ lies within the last
length-$w$ window, $t_e^j > n-w+1$.  A
straightforward solution is to profile the reuse time 
as before but only for the last $w-1$ accesses,
$\textit{ai}(n-w+2\dots n)$.
We denote the first and last accesses in the sub-trace by $\ftea$ and $\ltea$.
For an element $e$,
$\ftea$ equals to $t_e^j-(n-w+1)$ in the third case.
If $\fte>n-w+1$, we have $\ftea = \fte-(n-w+1)$ and
add $\fte-(n-w+1)$ to
$\fp_3(w)$. 
Using the reuse time histogram of the sub-trace,
$\fp_3(w)$
can be calculate as:
\[\fp_3(w) =  \sum_{i=1}^{w} i\times \rt'(i) + \sum_{e=1}^{m} \ftea-\sum_{n-w+1<\fte}(\fte-(n-w+1))\]

We actually do not need to profile again.
As a property of the reuse time, we have $\sum\limits_{i=1}^{n-1}i\times
\textit{rt}(i) = \sum\limits_{e=1}^m (\lte-\fte)$. Then
$\fp_3(w)$ equals to:
{\setlength\arraycolsep{2pt}
\begin{eqnarray}
\fp_3(w) &=&  \sum_{e=1}^{m} \ltea -\sum_{n-w+1<\fte}(\fte-(n-w+1)) \nonumber\\
&=&\sum_{n-w+1<\lte} \ltea-\sum_{n-w+1<\fte}(\fte-(n-w+1))\nonumber\\
& = &  \sum_{n-w+1<\lte} \left(\lte- (n-w+1)\right) -\sum_{n-w+1<\fte}(\fte-(n-w+1)) \nonumber
\end{eqnarray}}

Putting it all together, the final formula is:
{\setlength\arraycolsep{2pt}
\begin{eqnarray}\label{eq-singleside}
(n-w+1) \fp(w) =  \fp_1(w) + \fp_2(w) - \fp_3(w) =  wm + \sum_{i=1}^{n}\min(i,w)\times \textit{rt}(i) \nonumber\\
-   \sum_{e=1}^m d (w-\fte) - \sum_{e=1}^m d\left(\lte- (n-w+1)\right) 
\end{eqnarray}}

There is another explanation for the first and the last terms
in Equation~\ref{eq-singleside}. $wm$ means
that the first access of every element contributes $w$ to the footprint,
but for $\fte<w$, it only appears in $m-(w-\fte)$ windows.

The additive formula shows the WS-FP equivalence directly. When $n$ is infinite in the additive formula, all the terms except for $\fp_2$ can be omitted and the additive formula is equivalent to Denning-Schwartz.


\subsubsection{The Incremental Formula}

We say that a footprint calculation $\textit{fp}(w)$ is incremental
if it uses just the part of the reuse histogram $\textit{rt}(i)$ for $i
\le w$.  The Xiang formula and the additive method are not
incremental because they require the full reuse-time histogram
to compute any non-trivial footprint.

To obtain an incremental solution, we start from the initial estimate
that every window of size $w$ contains $w$ distinct elements.  The
maximal sum of working-set sizes is then $(n-w+1)w$.  In the following derivation, we
divide an access sequence into three parts: the head
$\textit{ai}(1\dots w-1)$, the body $\textit{ai}(w\dots n-w+1)$, and
the tail $\textit{ai}(n-w+2\dots n)$.

%

We now  
decreased the initial estimate $(n-w+1)w$
by removing the duplicates in all windows in each part.
Let's first consider the body
and assume that 
$t_e^{j_1-1} < w \le t_e^{j_1}<\dots<t_e^{j_2-1}\le n-w+1<t_e^{j_2}$
for element $e$. 
Consider a case where $j$ satisfies $j_1\le j-1<j\le j_2-1$.
The two accesses $\textit{ai}(t_e^{j-1})$ and $\textit{ai}(t_e^{j})$ 
appear in $d(w-rt_e^{j})$ windows.
All accesses of $e$ in the body decrease the initial estimate
by $d(w-rt_e^{j_1+1})+\dots+d(w-rt_e^{j_2-1})$.

In the head sequence, $\textit{ai}(1\dots w-1)$,
if $t_e^{j-1} < w \le t_e^{j}$,
the first $j-2$ accesses duplicate
$t_e^1+t_e^2+\dots+t_e^{j-2}$ times in the first $w$ windows,
and the $(j-1)$th access duplicates $d(w-rt_e^{j})$ times.
Similarly in the tail $\textit{ai}(n-w+2\dots n)$,
if $t_e^{j-1} \le n-w+1 < t_e^{j}$,
accesses of $e$ decrease the initial estimate by 
$d(w-rt_e^{j})+(w-t_e^{j+1})+\dots+(w-t_e^k)$.

The processing is shown in Algorithm~\ref{alg-performance-rt2} for the
head and the tail of a memory trace.  The algorithm adds $t_e^i$ or
$(w-t_e^j)$ and subtracts $rt$ in each part.  It requires specialized
information collection and has no succinct (mathematical)
representation except for the algorithm.  Given their results as
$\textit{lhead}, \textit{ltail}$, the complete formula is:
\[(n-w+1)\textit{fp}(w) = (n-w+1)w  - \sum_{i=1}^{w-1}(w-i)\textit{rt}(i)+ lhead + ltail\]

\begin{algorithm}
\caption{The Head/Tail Processing of the Incremental Method }
\label{alg-performance-rt2}

 $\lt(m)\leftarrow \{0\}$\;
\For{$i=1\TO w-1$}{
\If{$\lt(\textit{ai}(i))\neq 0$}{
 $lhead \leftarrow lhead - \lt(\textit{ai}(i)) + (i-\lt(\textit{ai}(i)))$\;
}
 $\lt(\textit{ai}(i))\leftarrow i$\;
}

 $\lt(m)\leftarrow \{0\}$\;
\For{$i=n-w+2 \TO n$}{
\If{$\lt(\textit{ai}(i)) \neq 0$}{
 $ltail \leftarrow ltail - (w-i) + (i-\lt(\textit{ai}(i)))$\;
}
 $\lt(\textit{ai}(i))\leftarrow i$\;
}

\end{algorithm}

To obtain a mathematical description, 
we use the sub-formula from the additive method to calculate the
footprint for windows in the head and tail parts directly.

First, we re-calculate the footprint just for the body part.
For each access in $\textit{ai}(w\dots n-w+1)$,
we start by assuming that every access contributes $w$ to the footprint
and obtain the initial estimate $(n-2w+2)w$.
Assume $e$'s reuse time sequence is
$t_e^{j_1-1} < w \le t_e^{j_1}<\dots<t_e^{j_2-1}\le n-w+1<t_e^{j_2}$. 
Based on the previous explanation,
the estimate is decreased by $d(w-rt_e^{j_1+1})+\dots+d(w-rt_e^{j_2-1})$.

For each access in the head $\textit{ai}(1\dots w-1)$ and tail 
$\textit{ai}(n-w+2\dots n)$,
we use the RT sequence algorithm 
to traverse either of them and denote the results
as $\fthead$,$\rthead$,$\lthead$ and $\fttail$,$\rttail$,$\lttail$ respectively.
If an element $e$ is accessed in $\textit{ai}(1\dots w-1)$,
we have $\ltehead = t_e^{j_1-1}$.
So $\textit{ai}((\ltehead$ contributes $\ltehead - d(w-rt_e^{j_1})$
to the footprint.
Similarly, if $e$ is accessed in $\textit{ai}(n-w+2\dots n)$, and
$\textit{ai}(\ftetail)$ contributes $(w-\ftetail)-d(w-rt_e^{j_2})$.
Let the contributions of the head and the tail
be $lhead'$ and $ltail'$.  They can be calculated as:
{\setlength\arraycolsep{2pt}
\begin{eqnarray}
lhead' &=& \sum_{e=1}^m\ltehead +  \sum_{i=1}^{w-1} (w-i)\rthead(i) = \sum_{e=1}^m\ftehead +  w\sum_{i=1}^{w-1} \rthead(i) \nonumber\\
ltail' &=& \sum_{e=1}^m(w-\ftetail)+  \sum_{i=1}^{w-1} (w-i)\rttail(i) =  w(w-1)-\sum_{e=1}^m\ltetail\nonumber
\end{eqnarray}}

Putting it all together, the incremental method is:
{\setlength\arraycolsep{2pt}
\begin{eqnarray}\label{eq-doubleside2}
(n-w+1)\textit{fp}(w) &=& (n-2w+2)w  - \sum_{i=1}^{w-1}(w-i)\textit{rt}(i)+ lhead' + ltail' \nonumber\\
&=& (n+1)m +(n-2m)w - \sum_{i=1}^{w-1}(w-i)\textit{rt}(i)  \nonumber\\
&&+ \sum_{e=1}^{m}\min(\fte,w) -\sum_{e=1}^{m}\max(\lte,n-w+1)
\end{eqnarray}}

The incremental method computes the footprint in $\bigo(w)$ time and $\bigo(w)$ space.
When $w\ll m$,
it has a significant advantage in efficiency over all previous solutions.

\subsection{Frequency Locality}
\label{sec:freq}

Frequency is concise ---
for any $n$ accesses to $m$ data, the average access frequency per
datum is $n/m$, a single number.  

It is commonly known as ``hotness''~\cite{Chilimbi+:PLDI99s,Rubin+:POPL02}.  Program data with a greater
number of reuses are hotter.  The locality is better if the
``temperature'' is higher.  
However, the ratio completely ignores the
order of data access.  The following three traces have the same access
frequency but different locality.  We name the first two following~\cite{DenningK:SOSP75} and the last one following~\cite{DingK:JPDC04}.
{\setlength\arraycolsep{0.3pt}
\begin{eqnarray}
\textit{cyclic} &:& e_1,e_2,\dots,e_m,e_1,e_2,\dots,e_m\nonumber\\
\textit{sawtooth} &:& e_1,e_2,\dots,e_m,e_m,\dots,e_2,e_1\nonumber\\
\textit{fused} &:& e_1,e_1,e_2,e_2,\dots,\dots,e_m,e_m\nonumber
\end{eqnarray}}
The locality depends on not just the frequency but also the recency of
reuse.  Although the three traces reuse the same data, the locality of
\emph{fused} is better than \emph{sawtooth}, and \emph{sawtooth}
better than \emph{cyclic}.  The closer the reuse is, the better the
locality.  

In theory, Snir and Yu showed that
the complete locality cannot be captured by a fixed size representation~\cite{SnirY:locality05}.
One way to measure locality is $\mr(c)$ for all $c \ge 0$.  The Snir-Yu limit implies that the frequency conversion has lost too much information --- it is impossible to compute the miss ratio from a fixed number of access frequencies.  

Not all locality definitions are equally usable. For the example, the \emph{fused} sequence has optimal locality, because no other access order can further reduce any reuse distance.  This optimality is obvious when analyzed using reuse distance but not using footprint or miss ratio.  Since different locality definitions are related, we can now take the optimal locality in one metric, e.g. reuse distance, and derive the optimal values of other metrics, e.g. footprint and miss ratio.  The next section presents the complete conversion theory.

\subsection{The Complete Theory}


Figure~\ref{fig:ctl} shows the \emph{MTL graph}, where each node is a metric of locality, each directed edge a conversion and, if the edge has a cross ($\times$), the assertion that no such conversion exists.  An undirected edge means two directed edge in opposite directions. The MTL conversions are injective.  A series of directed edges form a path.  The transitive relation gives the conversion or its impossibility between every pair of metrics.

	\begin{figure*}[h]
	\begin{center}
	   \includegraphics[width=0.99\textwidth]{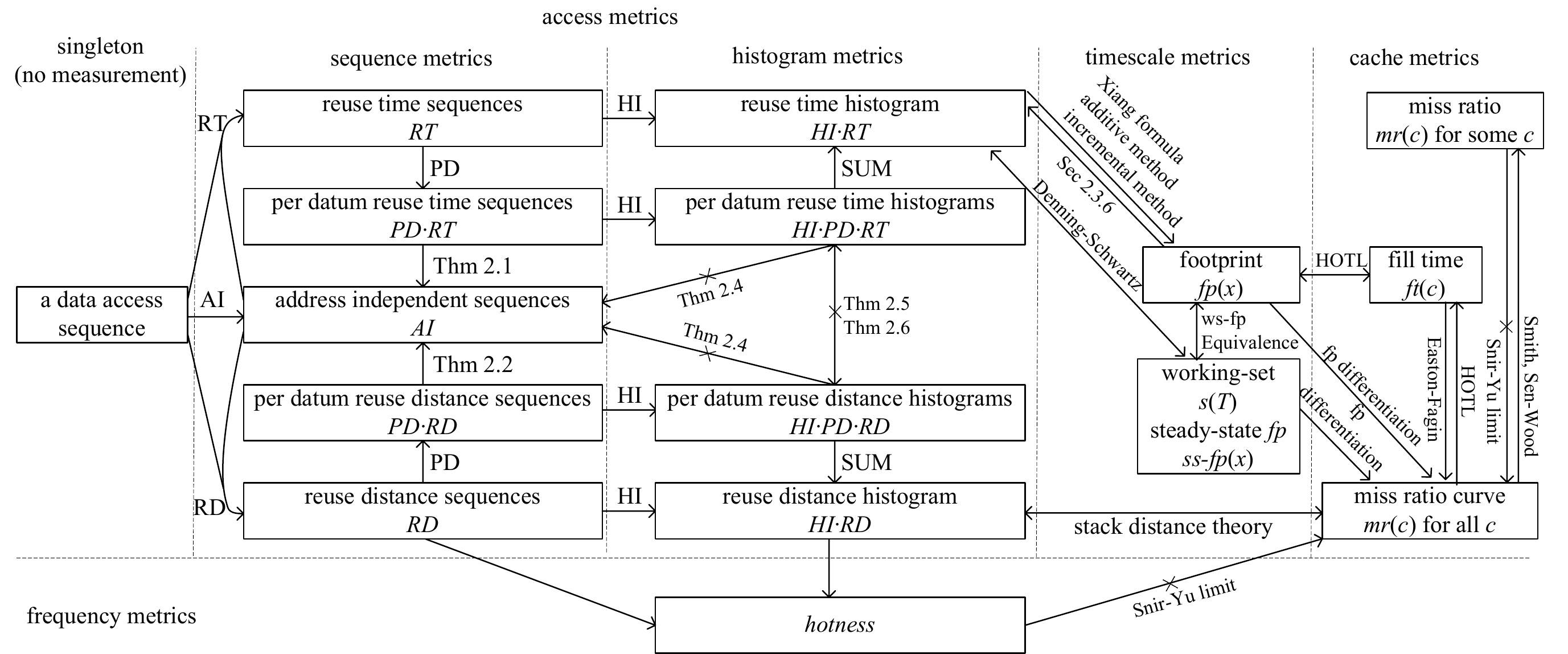}
	\end{center}
	   \caption{The conversion between the metrics of MTL}
	   \label{fig:ctl}
	\end{figure*}
    


The locality metrics are grouped by categories (Figure~\ref{fig:ct-tree}) into six areas in the MTL graph separated by dotted lines.  Timescale locality is centrally connected: it is the hub that connects histogram and cache metrics and through them, all metrics.  

All the metrics in the MTL graph are from existing work.  The contribution of the preceding sections is the connection of these metrics in particular the conversion and non-equivalence results that required for all-to-all relations and were absent from past work. 



\subsection{Usefulness in Practice}

The measurement theory helps to solve problems in practice. The first is precision. All metrics in the MTL graph are defined by mathematics or algorithms.  The second is concision and completeness in their relations.  A metric may be computed in different ways, and this is shown by multiple paths from the root.  Every derivation between two locality concepts is represented by a path in the graph.  Each conversion (edge) in the MTL graph is indivisible within itself, i.e. atomic. The third is modularity.  A path decomposes a complex construction into individual steps.  When there are multiple paths to derive the same metric, their overlap shows shared intermediate concepts and steps. These combinatorial choices are fully expressed but without being enumerated.


Mathematics is not just precise but maintains the precision after many steps of derivation.  Furthermore, it proves results for \emph{all} programs, which are therefore universal.  For example, for all programs, it takes linear time to computed the miss ratio of \emph{all} cache sizes,
and the computed miss ratio is bounded and monotone (Theorem~\ref{thm:fp-mr}).

Researchers can use multiple metrics when solving a problem. As the example in Section~\ref{sec:freq} shows, it is often convenient to formulate a problem using one metric and solution in another.  The measurement theory gives researchers full freedom in using these concepts in practice.  Its equivalence and conversion results provide safe bridges, and non-equivalence results mark the boundary and limitations.

\section{Related Work}
\label{sec:rel}

We review more related work in more detail in the following three areas: 

\paragraph{Timescale Locality}
In 1972, \citeauthor{DenningS:CACM72} gave a linear-time, iterative formula to compute the average working-set size from reuse times (inter-reference intervals). The derivation was based on stochastic assumptions --- that a trace is infinite and generated by a stationary Markov process, i.e. a limit value exists.  Later work extended the formula and used it on finite-length traces but did not extend the original theory~\citep{SlutzT:CACM74,DenningS:CACM78}.  The equivalence theorem (Theorem~\ref{thm-ws-fp-eq}) provides a theoretical justification why the Denning-Schwartz formula is accurate without stochastic assumptions. 

Shen was the main inventor of a formula that converts from reuse time to reuse distance statistically~\citeyear{Shen+:POPL07}.  Given the reuse time histogram, the Shen formula predicts the most likely reuse distance histogram.  The conversion was 99\% accurate and used by the open-source programming tool SLO~\cite{BeylsD:HPCC06}.  The conversion has many steps.  It was difficult to understand the reason for its accuracy.  The authors actually admitted in the paper that their ``formula is hard to interpret intuitively.''

\citet{Shen+:POPL07} defined $p(w)$ as ``the probability of any given data element to appear in a time interval of length'' $w$ and computed as follows from the reuse time histogram:


\[p(w)=\sum_{i=1}^{w}\sum_{j=i+1}^{n-1}\frac{\textit{rt}(j)}{m-1}\]

\noindent If we take the difference $p(w+1)-p(w)$, we see that it is equivalent to the Denning-Schwartz formula divided by $m-1$:

\[p(w+1)-p(w)=\sum_{i=w+2}^{n}\frac{\textit{rt}(i)}{m-1}\]

\noindent From footprint, the probability is $p(w)=\fp(w)/(m-1)$.  From its concavity (Corollary~\ref{cor:fp-shape}), we can easily prove that $\fp(w)/m$ is bounded between 0 and 1, as it should being a proper value of probability. The reason for $m-1$ is to model the reuse distance.  The reused datum cannot be accessed inside a reuse window.  The Shen formula can predict the reuse distance histogram accurately for many applications, which shows the accuracy of timescale locality. 

In 2010, \citeauthor{EklovH:ISPASS10} developed Statcache and showed that it was highly accurate (98\%) for computer-architecture evaluation. Statcache estimates the average reuse distance $ES(r)$ of all the accesses with the same reuse time $r$.
\citet{EklovH:ISPASS10} defined $F_j$ as the fraction of all memory references with a reuse time greater than $j$ and computed the average reuse distance $ES(r)$ using the following formula: 

\[ES(r) =\sum\limits_{j=1}^rF_j= \sum\limits_{j=1}^r\sum\limits_{i=j+1}^{n}\textit{rt}(i)\]

\noindent The purpose and the method of Statcache are similar to Shen.  While the basic formula is identical to Denning-Schwartz, Statcache also developed extremely fast measurement through a novel type of random sampling~\cite{EklovH:ISPASS10}. The subsequent application of Statstack won a best paper award a year later for its efficiency and accuracy~\cite{Eklov+:HiPEAC11}, which are the benefits of using timescale locality to model CPU caches.


The correctness of footprint differentiation was initially validated on the CPU
cache~\cite{Xiang+:PACT11,Xiang+:ASPLOS13}.  Three studies since 2014
further evaluated it for fully-associative LRU cache: memory access
in data cache~\cite{Wang+:CCGrid15}, object access in key-value cache,
i.e.  Memcached~\cite{Hu+:USENIX15}, and disk access in server
cache~\cite{Wires+:OSDI14}.  The three studies re-implemented the
footprint analysis independently and reported high accuracy through
extensive testing.  Hu et al. showed superior speed of convergence
using the theory~\cite{Hu+:USENIX15}.  Wang et al. showed strong
correlation (coefficient 0.938) between the predicted miss ratio and
measured co-run speed~\cite{Wang+:CCGrid15}.  

\paragraph{Shared Cache Modeling}
Published in 1978, the work of \citeauthor{EastonF:CACM78} was among the first to model the effect of cache sharing due to time sharing.  The paper explains the recipe intuitively but ``without proof''. Section~\ref{sec:recipe} derives the recipe mathematically from the measurement theory.  A technique to model shared cache is concurrent reuse distance, which shows the locality of a parallel execution precisely but does not have the property of composition as timescale metrics do~\cite{Schuff+:PACT10,WuY:TCS13}.  Many other techniques are hybrids where the locality is by reuse distance and the interference is by footprint~\cite{Xiang+:PPOPP11,ChenA:HPCA09,,Suh+:ICS01}, including one of the first models of multicore cache~\cite{Chandra+:HPCA05}.  The relations among these three types of models, footprint, reuse distance, and hybrid,  are explained by the measurement theory.





\paragraph{Cache Benchmark Synthesis}
Benchmark synthesis is the construction of a synthetic program with desirable locality.  It is locality metric conversion in the opposite direction to a trace. Synthesis has been used to solve two practical problems.  The first is memory probing with parameterized locality to examine machine performance in multiple use scenarios.  APEX-MAP is such a probe program that can be configured so its execution exhibits a distribution of reuse distances similar to a given target~\citep{IbrahimS:ICPP10}.  While APEX-MAP approximates, an algorithm by \citet{ShenS:LCPC08} generates a trace that has the exact reuse distance histogram as specified. The second use of synthesis is black-box benchmark cloning, for which cache behavior cloning is a sub-problem. A system called WEST generates a stochastic trace based on the reuse distance distribution within each cache set, in order to accurately replicate the behavior of set-associative caches~\cite{BalakrishnanS:HPCA12}.  




\section{Summary}

This paper has formally defined major metrics of locality, grouped them into six categories, and showed a series of relations and properties, including the equivalence between sequence metrics, non-equivalence between histogram metrics, the equivalence between two timescale metrics, a formal justification of the Easton-Fagin recipe, the first solution that computes all miss ratios from the footprint in linear time, and from these results, a complete measurement theory of locality. 

\appendix
\section*{APPENDIX}
\setcounter{section}{1}

{Non-equivalence of Histogram Metrics} In the following three theorems, counter examples are used to disprove equivalence.  

	\begin{theorem} \label{thm-rdrt2trace}
The memory trace cannot be built from its RT histogram or RD histogram,
or the histograms of individual elements.
	\end{theorem}
	\begin{proof}
The following two AI traces are different but have the same reuse distance histogram and reuse time histogram for the whole trace and for individual elements:
{\setlength\arraycolsep{0.3pt}
\begin{eqnarray}
ai&:& e_1,e_2,e_1,e_2,e_2,e_1\;\;\;\;\;\;ai': e_1,e_2,e_2,e_1,e_2,e_1\nonumber\\
rt&:& \infty,\infty,\;2,\;2,\;1,\;3\;\;\;\;\;\;\;\;rt': \infty,\infty,\;1,\;3,\;2,\;2\nonumber\\
rd&:& \infty,\infty,\;2,\;2,\;1,\;2\;\;\;\;\;\;\;\;\textit{rd}': \infty,\infty,\;1,\;2,\;2,\;2\nonumber
\end{eqnarray}}
\end{proof}

\ULforem 	
		\begin{theorem} \label{thm-rt2rd}
The RD histogram cannot be built from the RT histogram of the whole trace or individual elements.
	\end{theorem}
	
	\begin{proof}
The following two memory traces
produce the same reuse time histogram but different reuse distance histograms, where
accesses to $e_2,e_4$ are marked by $\uwave{e_2}, \uline{e_4}$ and the change of their location by $\uuline{e_2}, \uuline{e_4}$:

{\setlength\arraycolsep{2pt}
\small
\begin{eqnarray}
ai&:& e_1,\uwave{e_2},e_3,\uline{e_4},e_3,\uuline{e_4},e_1,\uwave{e_2},e_3,\uline{e_4},e_3,\uwave{e_2},e_3,\uuline{e_2},e_3,\uline{e_4},e_3,\uwave{e_2},e_1\nonumber\\
rt&:& \infty,\infty,\infty,\infty,\,2,\phantom{_1}2,\phantom{_1}6,\phantom{_1}6,\phantom{_1}4,\phantom{_1}4,\phantom{_1}2,\phantom{_1}4,\phantom{_1}2,\phantom{_1}2,\phantom{_1}2,\phantom{_1}6,2,\phantom{_1}4,12\nonumber\\
rd&:&  \infty,\infty,\infty,\infty,\,2,\phantom{_1}2,\phantom{_1}4,\phantom{_1}4,\phantom{_1}4,\phantom{_1}4,\phantom{_1}2,\phantom{_1}3,\phantom{_1}2,\phantom{_1}2,\phantom{_1}2,\phantom{_1}3,2,\phantom{_1}3,\phantom{_1}4\nonumber\nonumber
\end{eqnarray}}
{\setlength\arraycolsep{2pt}
\addtolength{\abovedisplayskip}{-5mm}
\small
\begin{eqnarray}
ai'&:& e_1,\uwave{e_2},e_3,\uline{e_4},e_3,\uuline{e_2},e_1,\uwave{e_2},e_3,\uline{e_4},e_3,\uwave{e_2},e_3,\uuline{e_4},e_3,\uline{e_4},e_3,\uwave{e_2},e_1\nonumber\\
rt'&:& \infty,\infty,\infty,\infty,\,2,\phantom{_1}4,\phantom{_1}6,\phantom{_1}2,\phantom{_1}4,\phantom{_1}6,\phantom{_1}2,\phantom{_1}4,\phantom{_1}2,\phantom{_1}4,\phantom{_1}2,\phantom{_1}2,2,\phantom{_1}6,12\nonumber\\
rd'&:&\infty,\infty,\infty,\infty,\,2,\phantom{_1}3,\phantom{_1}4,\phantom{_1}2,\phantom{_1}3,\phantom{_1}4,\phantom{_1}2,\phantom{_1}3,\phantom{_1}2,\phantom{_1}3,\phantom{_1}2,\phantom{_1}2,2,\phantom{_1}3,\phantom{_1}4\nonumber\nonumber
\end{eqnarray}}
\end{proof}

\normalem

		\begin{theorem} \label{thm-rd2rt}
The RT histogram cannot be built from the RD histogram of the whole trace or individual elements.
	\end{theorem}
	
	\begin{proof}
The following two memory traces
produce the same reuse distance histogram but different reuse time histograms:
{\setlength\arraycolsep{2pt}
\begin{eqnarray}
ai&:& e_1,e_2,e_3,e_4,e_3,e_4,\uuline{e_1},\uuline{e_2},e_3,e_4,e_3,e_2,e_1\nonumber\\
rt&:& \infty,\infty,\infty,\infty,\,2,\phantom{_1}2,\phantom{_1}6,\phantom{_1}6,\phantom{_1}4,\phantom{_1}4,\phantom{_1}2,\phantom{_1}4,12\nonumber\\
rd&:&  \infty,\infty,\infty,\infty,\,2,\phantom{_1}2,\phantom{_1}4,\phantom{_1}4,\phantom{_1}4,\phantom{_1}4,\phantom{_1}2,\phantom{_1}3,\phantom{_1}4\nonumber\nonumber
\end{eqnarray}}
{\setlength\arraycolsep{2pt}
\addtolength{\abovedisplayskip}{-5mm}
\begin{eqnarray}
ai'&:& e_1,e_2,e_3,e_4,e_3,e_4,\uuline{e_2},\uuline{e_1},e_3,e_4,e_3,e_2,e_1\nonumber\\
rt'&:& \infty,\infty,\infty,\infty,\,2,\phantom{_1}2,\phantom{_1}5,\phantom{_1}7,\phantom{_1}4,\phantom{_1}4,\phantom{_1}2,\phantom{_1}5,11\nonumber\\
rd'&:& \infty,\infty,\infty,\infty,\,2,\phantom{_1}2,\phantom{_1}3,\phantom{_1}4,\phantom{_1}4,\phantom{_1}4,\phantom{_1}2,\phantom{_1}4,\phantom{_1}4\nonumber\nonumber
\end{eqnarray}}
\end{proof}

\section{lemma \ref{lemma-2} to theorem \ref{thm-normal-rt}}
\label{appendix-3}
{\setlength\arraycolsep{2pt}
\begin{eqnarray}
\sum\limits_{i=1}^{n-1}i\times \textit{rt}(i) &=& \sum\limits_{e=1}^m (\lte-\fte)\nonumber\\
&=& \sum\limits_{e=1}^m (\lte-(n-m)+(n-m) -\fte)\nonumber\\
&=& (n-m)\times m + \sum\limits_{e=1}^m (\lte-(n-m))-\sum\limits_{e=1}^m \fte \nonumber\\
&=& (n-m)\times m + \sum\limits_{i=1}^m i-\sum\limits_{i=1}^m i \nonumber\\
&=& (n-m)\times m\nonumber
\end{eqnarray}}

\section{Eq. (\ref{eq-nonside}) to Eq. (\ref{eq-doubleside2})}
\label{appendix-1}

{\setlength\arraycolsep{2pt}
\begin{eqnarray}
&&(n-w+1)\textit{fp}(w)\nonumber\\
& = & (n-w+1)m  -  \sum_{i=w+1}^{n}(i-w)\textit{rt}(i)\nonumber\\
&&- \sum_{e=1}^{m}d(\fte-w)-  \sum_{e=1}^{m}d(n-w+1-\lte)\nonumber\\
&=&  (n-w+1)m - \sum_{i=w+1}^{n}i\times \textit{rt}(i)+w\sum_{i=w+1}^{n}\textit{rt}(i)\nonumber\\
&&- \sum_{e=1}^{m}d(\fte-w)-  \sum_{e=1}^{m}d(n-w+1-\lte)\nonumber\\
&=&  (n-w+1)m - \sum_{i=1}^{n}i\times \textit{rt}(i)+\sum_{i=1}^{w}i\times \textit{rt}(i)\nonumber\\
&&+  w\sum_{i=1}^{n}\textit{rt}(i)- w\sum_{i=1}^{w}\textit{rt}(i)\nonumber\\
&&- \sum_{e=1}^{m}d(\fte-w)-  \sum_{e=1}^{m}d(n-w+1-\lte)\nonumber\\
&=&  (n-w+1)m +\sum_{i=1}^{w}i\times \textit{rt}(i)- w\sum_{i=1}^{w}\textit{rt}(i)\nonumber\\
&&+  w(n-m) - \sum_{e=1}^m (\lte-\fte)\nonumber\\
&&- \sum_{e=1}^{m}d(\fte-w)-  \sum_{e=1}^{m}d(n-w+1-\lte)\nonumber\\
&=& (n+1)m +(n-2m)w - \sum_{i=1}^{w-1}(w-i)\textit{rt}(i)\nonumber\\
&&+ \sum_{e=1}^{m}\min(\fte,w) -  \sum_{e=1}^{m}\max(\lte,n-w+1)\nonumber
\end{eqnarray}}

\section{Eq. (\ref{eq-nonside}) to Eq. (\ref{eq-singleside})}
\label{appendix-2}

{\setlength\arraycolsep{2pt}
\begin{eqnarray}
&&(n-w+1)\textit{fp}(w)\nonumber\\
& = & (n-w+1)m  -  \sum_{i=w+1}^{n}(i-w)\textit{rt}(i)\nonumber\\
&&- \sum_{e=1}^{m}d(\fte-w)-  \sum_{e=1}^{m}d(n-w+1-\lte)\nonumber\\
& = & (n-w+1)m  -  \sum_{i=w+1}^{n}(i-w)\textit{rt}(i)\nonumber\\
&&- \sum_{e=1}^{m}(\fte-w)-  \sum_{e=1}^{m}(n-w+1-\lte)\nonumber\\
&&+ \sum_{\fte< w}(\fte-w)+  \sum_{n-w+1< \lte}(n-w+1-\lte)\nonumber\\
& = & (n-w+1)m  -  \sum_{i=w+1}^{n}(i-w)\textit{rt}(i)\nonumber\\
&&+ \sum_{e=1}^{m}(\lte - \fte)+ mw-  m(n-w+1) \nonumber\\
&&+ \sum_{\fte< w}(\fte-w)+\sum_{n-w+1< \lte}(n-w+1-\lte)\nonumber\\
& = &  w\sum_{i=w+1}^{n}\textit{rt}(i)   -\sum_{i=w+1}^{n}i\times \textit{rt}(i)+\sum_{i=1}^{n}i\times \textit{rt}(i)  + mw\nonumber\\
&&+ \sum_{\fte< w}(\fte-w)+  \sum_{n-w+1< \lte}(n-w+1-\lte)\nonumber\\
& = &  \sum_{i=1}^{n}\min(i,w)\textit{rt}(i) + mw \nonumber\\
&&- \sum_{e=1}^m d(w-\fte)- \sum_{e=1}^m d(\lte- (n-w+1))\nonumber
\end{eqnarray}}

\ULforem 	
\bibliographystyle{ACM-Reference-Format-Journals}
\bibliography{all}


\begin{thebibliography}{00}


\ifx \showCODEN    \undefined \def \showCODEN     #1{\unskip}     \fi
\ifx \showDOI      \undefined \def \showDOI       #1{{\tt DOI:}\penalty0{#1}\ }
  \fi
\ifx \showISBNx    \undefined \def \showISBNx     #1{\unskip}     \fi
\ifx \showISBNxiii \undefined \def \showISBNxiii  #1{\unskip}     \fi
\ifx \showISSN     \undefined \def \showISSN      #1{\unskip}     \fi
\ifx \showLCCN     \undefined \def \showLCCN      #1{\unskip}     \fi
\ifx \shownote     \undefined \def \shownote      #1{#1}          \fi
\ifx \showarticletitle \undefined \def \showarticletitle #1{#1}   \fi
\ifx \showURL      \undefined \def \showURL       #1{#1}          \fi

\bibitem[\protect\citeauthoryear{Balakrishnan and Solihin}{Balakrishnan and
  Solihin}{2012}]%
        {BalakrishnanS:HPCA12}
{Ganesh Balakrishnan} {and} {Yan Solihin}. 2012.
\newblock \showarticletitle{{WEST:} Cloning data cache behavior using
  Stochastic Traces}. In {\em Proceedings of HPCA}. 387--398.
\newblock
\showDOI{%
\url{http://dx.doi.org/10.1109/HPCA.2012.6169042}}


\bibitem[\protect\citeauthoryear{Batson and Madison}{Batson and
  Madison}{1976}]%
        {BatsonM:SIGMETRICS76}
{A.~P. Batson} {and} {A.~W. Madison}. 1976.
\newblock \showarticletitle{Measurements of major locality phases in symbolic
  reference strings}. In {\em Proceedings of SIGMETRICS}. Cambridge, MA.
\newblock


\bibitem[\protect\citeauthoryear{Belady}{Belady}{1966}]%
        {Belady66}
{L.~A. Belady}. 1966.
\newblock \showarticletitle{A study of replacement algorithms for a
  virtual-storage computer}.
\newblock {\em IBM Systems Journal\/} {5}, 2 (1966), 78--101.
\newblock


\bibitem[\protect\citeauthoryear{Beyls and D'Hollander}{Beyls and
  D'Hollander}{2006}]%
        {BeylsD:HPCC06}
{Kristof Beyls} {and} {Erik~H. D'Hollander}. 2006.
\newblock \showarticletitle{Discovery of locality-improving refactoring by
  reuse path analysis}. In {\em Proceedings of High Performance Computing and
  Communications. Springer. Lecture Notes in Computer Science}, Vol. 4208.
  220--229.
\newblock


\bibitem[\protect\citeauthoryear{Buzen}{Buzen}{2015}]%
        {Buzen:Book15}
{Jeffrey~P. Buzen}. 2015.
\newblock {\em Rethinking randomness: a new foundation for stochastic
  modeling}.
\newblock
\showISBNx{978-1-508-43598-3}


\bibitem[\protect\citeauthoryear{Chandra, Guo, Kim, and Solihin}{Chandra
  et~al\mbox{.}}{2005}]%
        {Chandra+:HPCA05}
{Dhruba Chandra}, {Fei Guo}, {Seongbeom Kim}, {and} {Yan Solihin}. 2005.
\newblock \showarticletitle{Predicting Inter-Thread Cache Contention on a Chip
  Multi-Processor Architecture}. In {\em Proceedings of HPCA}. 340--351.
\newblock


\bibitem[\protect\citeauthoryear{Chen and Aamodt}{Chen and Aamodt}{2009}]%
        {ChenA:HPCA09}
{Xi~E. Chen} {and} {Tor~M. Aamodt}. 2009.
\newblock \showarticletitle{A first-order fine-grained multithreaded throughput
  model}. In {\em Proceedings of HPCA}. 329--340.
\newblock


\bibitem[\protect\citeauthoryear{Chilimbi, Davidson, and Larus}{Chilimbi
  et~al\mbox{.}}{1999}]%
        {Chilimbi+:PLDI99s}
{Trishul~M. Chilimbi}, {Bob Davidson}, {and} {James~R. Larus}. 1999.
\newblock \showarticletitle{Cache-Conscious Structure Definition}. In {\em
  Proceedings of PLDI}. 13--24.
\newblock


\bibitem[\protect\citeauthoryear{Denning}{Denning}{1968}]%
        {Denning:CACM68}
{Peter~J. Denning}. 1968.
\newblock \showarticletitle{The working set model for program behaviour}.
\newblock {\it Commun. {ACM}} {11}, 5 (1968), 323--333.
\newblock


\bibitem[\protect\citeauthoryear{Denning}{Denning}{2005}]%
        {Denning:CACM05loc}
{Peter~J. Denning}. 2005.
\newblock \showarticletitle{The locality principle}.
\newblock {\it Commun. ACM} {48}, 7 (2005), 19--24.
\newblock


\bibitem[\protect\citeauthoryear{Denning and Buzen}{Denning and Buzen}{1978}]%
        {DenningB:CSUR78}
{Peter~J. Denning} {and} {Jeffrey~P. Buzen}. 1978.
\newblock \showarticletitle{The Operational Analysis of Queueing Network
  Models}.
\newblock {\em {ACM} Comput. Surv.\/} {10}, 3 (1978), 225--261.
\newblock
\showDOI{%
\url{http://dx.doi.org/10.1145/356733.356735}}


\bibitem[\protect\citeauthoryear{Denning and Kahn}{Denning and Kahn}{1975}]%
        {DenningK:SOSP75}
{Peter~J. Denning} {and} {Kevin~C. Kahn}. 1975.
\newblock \showarticletitle{A study of program locality and lifetime
  functions}. In {\em Proceedings of the ACM Symposium on Operating System
  Principles}. 207--216.
\newblock


\bibitem[\protect\citeauthoryear{Denning and Martell}{Denning and
  Martell}{2015}]%
        {DenningM:Book15}
{Peter~J. Denning} {and} {Craig~H. Martell}. 2015.
\newblock {\em Great Principles of Computing}.
\newblock MIT Press.
\newblock


\bibitem[\protect\citeauthoryear{Denning and Schwartz}{Denning and
  Schwartz}{1972}]%
        {DenningS:CACM72}
{Peter~J. Denning} {and} {Stuart~C. Schwartz}. 1972.
\newblock \showarticletitle{Properties of the working set model}.
\newblock {\it Commun. ACM} {15}, 3 (1972), 191--198.
\newblock


\bibitem[\protect\citeauthoryear{Denning and Slutz}{Denning and Slutz}{1978}]%
        {DenningS:CACM78}
{Peter~J. Denning} {and} {Donald~R. Slutz}. 1978.
\newblock \showarticletitle{Generalized working sets for segment reference
  strings}.
\newblock {\it Commun. ACM} {21}, 9 (1978), 750--759.
\newblock


\bibitem[\protect\citeauthoryear{Ding and Kennedy}{Ding and Kennedy}{2004}]%
        {DingK:JPDC04}
{C. Ding} {and} {K. Kennedy}. 2004.
\newblock \showarticletitle{Improving effective bandwidth through compiler
  enhancement of global cache reuse}.
\newblock {\it J. Parallel and Distrib. Comput.} {64}, 1 (2004), 108--134.
\newblock


\bibitem[\protect\citeauthoryear{Ding, Xiang, Bao, Luo, Luo, and Wang}{Ding
  et~al\mbox{.}}{2014}]%
        {Ding+:JCST14}
{Chen Ding}, {Xiaoya Xiang}, {Bin Bao}, {Hao Luo}, {Ying{-}Wei Luo}, {and}
  {Xiao{-}lin Wang}. 2014.
\newblock \showarticletitle{Performance Metrics and Models for Shared Cache}.
\newblock {\em J. Comput. Sci. Technol.\/} {29}, 4 (2014), 692--712.
\newblock
\showDOI{%
\url{http://dx.doi.org/10.1007/s11390-014-1460-7}}


\bibitem[\protect\citeauthoryear{Easton and Fagin}{Easton and Fagin}{1978}]%
        {EastonF:CACM78}
{Malcolm~C. Easton} {and} {Ronald Fagin}. 1978.
\newblock \showarticletitle{Cold-Start vs. Warm-Start Miss Ratios}.
\newblock {\it Commun. ACM} {21}, 10 (1978), 866--872.
\newblock


\bibitem[\protect\citeauthoryear{Eklov, Black-Schaffer, and Hagersten}{Eklov
  et~al\mbox{.}}{2011}]%
        {Eklov+:HiPEAC11}
{David Eklov}, {David Black-Schaffer}, {and} {Erik Hagersten}. 2011.
\newblock \showarticletitle{Fast modeling of shared caches in multicore
  systems}. In {\em Proceedings of HiPEAC}. 147--157.
\newblock
\newblock
\shownote{\emph{Best paper}.}


\bibitem[\protect\citeauthoryear{Eklov and Hagersten}{Eklov and
  Hagersten}{2010}]%
        {EklovH:ISPASS10}
{David Eklov} {and} {Erik Hagersten}. 2010.
\newblock \showarticletitle{{StatStack}: Efficient modeling of {LRU} caches}.
  In {\em Proceedings of ISPASS}. 55--65.
\newblock


\bibitem[\protect\citeauthoryear{Elango, Rastello, Pouchet, Ramanujam, and
  Sadayappan}{Elango et~al\mbox{.}}{2015}]%
        {Elango+:POPL15}
{Venmugil Elango}, {Fabrice Rastello}, {Louis{-}No{\"{e}}l Pouchet}, {J.
  Ramanujam}, {and} {P. Sadayappan}. 2015.
\newblock \showarticletitle{On Characterizing the Data Access Complexity of
  Programs}. In {\em Proceedings of POPL}. 567--580.
\newblock
\showDOI{%
\url{http://dx.doi.org/10.1145/2676726.2677010}}


\bibitem[\protect\citeauthoryear{Fang, Carr, {\"O}nder, and Wang}{Fang
  et~al\mbox{.}}{2005}]%
        {Fang+:PACT05}
{Changpeng Fang}, {Steve Carr}, {Soner {\"O}nder}, {and} {Zhenlin Wang}. 2005.
\newblock \showarticletitle{Instruction Based Memory Distance Analysis and its
  Application}. In {\em Proceedings of PACT}. 27--37.
\newblock


\bibitem[\protect\citeauthoryear{Gupta, Xiang, Yang, and Zhou}{Gupta
  et~al\mbox{.}}{2013}]%
        {Gupta+:JPDC13}
{Saurabh Gupta}, {Ping Xiang}, {Yi Yang}, {and} {Huiyang Zhou}. 2013.
\newblock \showarticletitle{Locality principle revisited: {A} probability-based
  quantitative approach}.
\newblock {\em JPDC\/} {73}, 7 (2013), 1011--1027.
\newblock
\showDOI{%
\url{http://dx.doi.org/10.1016/j.jpdc.2013.01.010}}


\bibitem[\protect\citeauthoryear{Hong and Kung}{Hong and Kung}{1981}]%
        {HongK:STOC81}
{J. Hong} {and} {H.~T. Kung}. 1981.
\newblock \showarticletitle{{I/O} complexity: The red-blue pebble game}. In
  {\em Proceedings of the ACM Conference on Theory of Computing}. Milwaukee,
  WI.
\newblock


\bibitem[\protect\citeauthoryear{Hu, Wang, Li, Zhou, Luo, Ding, Jiang, and
  Wang}{Hu et~al\mbox{.}}{2015}]%
        {Hu+:USENIX15}
{Xiameng Hu}, {Xiaolin Wang}, {Yechen Li}, {Lan Zhou}, {Yingwei Luo}, {Chen
  Ding}, {Song Jiang}, {and} {Zhenlin Wang}. 2015.
\newblock \showarticletitle{{LAMA}: Optimized Locality-aware Memory Allocation
  for Key-value Cache}. In {\em Proceedings of USENIX ATC}.
\newblock


\bibitem[\protect\citeauthoryear{Hu, Wang, Zhou, Luo, Ding, and Wang}{Hu
  et~al\mbox{.}}{2016}]%
        {Hu+:USENIX16}
{Xiameng Hu}, {Xiaolin Wang}, {Lan Zhou}, {Yingwei Luo}, {Chen Ding}, {and}
  {Zhenlin Wang}. 2016.
\newblock \showarticletitle{Kinetic Modeling of Data Eviction in Cache}. In
  {\em Proceedings of USENIX ATC}.
\newblock


\bibitem[\protect\citeauthoryear{Ibrahim and Strohmaier}{Ibrahim and
  Strohmaier}{2010}]%
        {IbrahimS:ICPP10}
{Khaled~Z. Ibrahim} {and} {Erich Strohmaier}. 2010.
\newblock \showarticletitle{Characterizing the Relation Between {Apex-Map}
  Synthetic Probes and Reuse Distance Distributions}.
\newblock {\em Proceedings of ICPP\/}  {0} (2010), 353--362.
\newblock
\showISSN{0190-3918}
\showDOI{%
\url{http://dx.doi.org/10.1109/ICPP.2010.43}}


\bibitem[\protect\citeauthoryear{Jiang and Zhang}{Jiang and Zhang}{2002}]%
        {JiangZ:SIGMetrics02}
{S. Jiang} {and} {X. Zhang}. 2002.
\newblock \showarticletitle{{LIRS}: an efficient low inter-reference recency
  set replacement to improve buffer cache performance}. In {\em Proceedings of
  SIGMETRICS}. Marina Del Rey, California.
\newblock


\bibitem[\protect\citeauthoryear{Lavaee}{Lavaee}{2016}]%
        {Lavaee:POPL16}
{Rahman Lavaee}. 2016.
\newblock \showarticletitle{The Hardness of Data Packing}. In {\em Proceedings
  of the 43rd Annual ACM SIGPLAN-SIGACT Symposium on Principles of Programming
  Languages} {\em (POPL 2016)}. ACM, New York, NY, USA, 232--242.
\newblock
\showISBNx{978-1-4503-3549-2}


\bibitem[\protect\citeauthoryear{Lu and Scott}{Lu and Scott}{2011}]%
        {LuS:DISC11}
{Li Lu} {and} {Michael~L. Scott}. 2011.
\newblock \showarticletitle{Toward a Formal Semantic Framework for
  Deterministic Parallel Programming}. In {\em Proceedings of the International
  Conference on Distributed Computing}. 460--474.
\newblock


\bibitem[\protect\citeauthoryear{Marin and Mellor-Crummey}{Marin and
  Mellor-Crummey}{2004}]%
        {MarinM:SIGMETRICS04}
{G. Marin} {and} {J. Mellor-Crummey}. 2004.
\newblock \showarticletitle{Cross architecture performance predictions for
  scientific applications using parameterized models}. In {\em Proceedings of
  SIGMETRICS}. 2--13.
\newblock


\bibitem[\protect\citeauthoryear{Mattson, Gecsei, Slutz, and Traiger}{Mattson
  et~al\mbox{.}}{1970}]%
        {Mattson+:IBM70}
{R.~L. Mattson}, {J. Gecsei}, {D. Slutz}, {and} {I.~L. Traiger}. 1970.
\newblock \showarticletitle{Evaluation techniques for storage hierarchies}.
\newblock {\em {IBM} System Journal\/} {9}, 2 (1970), 78--117.
\newblock


\bibitem[\protect\citeauthoryear{Nugteren, van~den Braak, Corporaal, and
  Bal}{Nugteren et~al\mbox{.}}{2014}]%
        {Nugteren+:HPCA14}
{Cedric Nugteren}, {Gert{-}Jan van~den Braak}, {Henk Corporaal}, {and}
  {Henri~E. Bal}. 2014.
\newblock \showarticletitle{A detailed {GPU} cache model based on reuse
  distance theory}. In {\em Proceedings of HPCA}.
\newblock


\bibitem[\protect\citeauthoryear{Petrank and Rawitz}{Petrank and
  Rawitz}{2002}]%
        {PetrankR:POPL02}
{E. Petrank} {and} {D. Rawitz}. 2002.
\newblock \showarticletitle{The Hardness of Cache Conscious Data Placement}. In
  {\em Proceedings of POPL}.
\newblock


\bibitem[\protect\citeauthoryear{Rubin, Bodik, and Chilimbi}{Rubin
  et~al\mbox{.}}{2002}]%
        {Rubin+:POPL02}
{S. Rubin}, {R. Bodik}, {and} {T. Chilimbi}. 2002.
\newblock \showarticletitle{An efficient profile-analysis framework for data
  layout optimizations}. In {\em Proceedings of POPL}. Portland, Oregon.
\newblock


\bibitem[\protect\citeauthoryear{Schuff, Kulkarni, and Pai}{Schuff
  et~al\mbox{.}}{2010}]%
        {Schuff+:PACT10}
{Derek~L. Schuff}, {Milind Kulkarni}, {and} {Vijay~S. Pai}. 2010.
\newblock \showarticletitle{Accelerating multicore reuse distance analysis with
  sampling and parallelization}. In {\em Proceedings of PACT}. 53--64.
\newblock


\bibitem[\protect\citeauthoryear{Sen and Wood}{Sen and Wood}{2013}]%
        {SenW:SIGMETRICS13}
{Rathijit Sen} {and} {David~A. Wood}. 2013.
\newblock \showarticletitle{Reuse-based online models for caches}. In {\em
  Proceedings of SIGMETRICS}. 279--292.
\newblock


\bibitem[\protect\citeauthoryear{Shen and Shaw}{Shen and Shaw}{2008}]%
        {ShenS:LCPC08}
{Xipeng Shen} {and} {Jonathan Shaw}. 2008.
\newblock \showarticletitle{Scalable Implementation of Efficient Locality
  Approximation}. In {\em Proceedings of the LCPC Workshop}. 202--216.
\newblock


\bibitem[\protect\citeauthoryear{Shen, Shaw, Meeker, and Ding}{Shen
  et~al\mbox{.}}{2007a}]%
        {Shen+:POPL07}
{Xipeng Shen}, {Jonathan Shaw}, {Brian Meeker}, {and} {Chen Ding}. 2007a.
\newblock \showarticletitle{Locality approximation using time}. In {\em
  Proceedings of POPL}. 55--61.
\newblock


\bibitem[\protect\citeauthoryear{Shen, Zhong, and Ding}{Shen
  et~al\mbox{.}}{2007b}]%
        {Shen+:JPDC07}
{X. Shen}, {Y. Zhong}, {and} {C. Ding}. 2007b.
\newblock \showarticletitle{Predicting locality phases for dynamic memory
  optimization}.
\newblock {\it J. Parallel and Distrib. Comput.} {67}, 7 (2007), 783--796.
\newblock


\bibitem[\protect\citeauthoryear{Slutz and Traiger}{Slutz and Traiger}{1974}]%
        {SlutzT:CACM74}
{Donald~R. Slutz} {and} {Irving~L. Traiger}. 1974.
\newblock \showarticletitle{A Note on the Calculation Working Set Size}.
\newblock {\it Commun. ACM} {17}, 10 (1974), 563--565.
\newblock
\showDOI{%
\url{http://dx.doi.org/10.1145/355620.361167}}


\bibitem[\protect\citeauthoryear{Smith}{Smith}{1976}]%
        {Smith:ICSE76}
{A.~J. Smith}. 1976.
\newblock \showarticletitle{On the Effectiveness of Set Associative Page
  Mapping and Its Applications in Main Memory Management}. In {\em Proceedings
  of ICSE}.
\newblock


\bibitem[\protect\citeauthoryear{Snir and Yu}{Snir and Yu}{2005}]%
        {SnirY:locality05}
{M. Snir} {and} {J. Yu}. 2005.
\newblock {\em On the theory of spatial and temporal locality}.
\newblock {T}echnical {R}eport DCS-R-2005-2564. Computer Science Dept., Univ.
  of Illinois at Urbana-Champaign.
\newblock


\bibitem[\protect\citeauthoryear{Suh, Devadas, and Rudolph}{Suh
  et~al\mbox{.}}{2001}]%
        {Suh+:ICS01}
{G.~Edward Suh}, {Srinivas Devadas}, {and} {Larry Rudolph}. 2001.
\newblock \showarticletitle{Analytical cache models with applications to cache
  partitioning.}. In {\em Proceedings of ICS}. 1--12.
\newblock


\bibitem[\protect\citeauthoryear{{Wang et al.}}{{Wang et al.}}{2015}]%
        {Wang+:CCGrid15}
{{Wang et al.}} 2015.
\newblock \showarticletitle{Optimal Program Symbiosis in Shared Cache}. In {\em
  Proceedings of CCGrid}.
\newblock


\bibitem[\protect\citeauthoryear{Wires, Ingram, Drudi, Harvey, Warfield, and
  Data}{Wires et~al\mbox{.}}{2014}]%
        {Wires+:OSDI14}
{Jake Wires}, {Stephen Ingram}, {Zachary Drudi}, {Nicholas~JA Harvey}, {Andrew
  Warfield}, {and} {Coho Data}. 2014.
\newblock \showarticletitle{Characterizing storage workloads with counter
  stacks}. In {\em Proceedings of OSDI}. USENIX Association, 335--349.
\newblock


\bibitem[\protect\citeauthoryear{Wu and Yeung}{Wu and Yeung}{2013}]%
        {WuY:TCS13}
{Meng{-}Ju Wu} {and} {Donald Yeung}. 2013.
\newblock \showarticletitle{Efficient Reuse Distance Analysis of Multicore
  Scaling for Loop-Based Parallel Programs}.
\newblock {\em {ACM} Trans. Comput. Syst.\/} {31}, 1 (2013), 1.
\newblock
\showDOI{%
\url{http://dx.doi.org/10.1145/2427631.2427632}}


\bibitem[\protect\citeauthoryear{Wu and Yeung}{Wu and Yeung}{2011}]%
        {WuY:PACT11}
{Meng-Ju Wu} {and} {Donald Yeung}. 2011.
\newblock \showarticletitle{Coherent Profiles: Enabling Efficient Reuse
  Distance Analysis of Multicore Scaling for Loop-based Parallel Programs}. In
  {\em Proceedings of PACT}. 264--275.
\newblock


\bibitem[\protect\citeauthoryear{Wu, Zhao, and Yeung}{Wu et~al\mbox{.}}{2013}]%
        {Wu+:ISCA13}
{Meng-Ju Wu}, {Minshu Zhao}, {and} {Donald Yeung}. 2013.
\newblock \showarticletitle{Studying multicore processor scaling via reuse
  distance analysis}. In {\em Proceedings of ISCA}. 499--510.
\newblock


\bibitem[\protect\citeauthoryear{Xiang, Bao, Bai, Ding, and Chilimbi}{Xiang
  et~al\mbox{.}}{2011a}]%
        {Xiang+:PPOPP11}
{Xiaoya Xiang}, {Bin Bao}, {Tongxin Bai}, {Chen Ding}, {and} {Trishul~M.
  Chilimbi}. 2011a.
\newblock \showarticletitle{All-window profiling and composable models of cache
  sharing}. In {\em Proceedings of PPoPP}. 91--102.
\newblock


\bibitem[\protect\citeauthoryear{Xiang, Bao, Ding, and Gao}{Xiang
  et~al\mbox{.}}{2011b}]%
        {Xiang+:PACT11}
{Xiaoya Xiang}, {Bin Bao}, {Chen Ding}, {and} {Yaoqing Gao}. 2011b.
\newblock \showarticletitle{Linear-time Modeling of Program Working Set in
  Shared Cache}. In {\em Proceedings of PACT}. 350--360.
\newblock


\bibitem[\protect\citeauthoryear{Xiang, Ding, Luo, and Bao}{Xiang
  et~al\mbox{.}}{2013}]%
        {Xiang+:ASPLOS13}
{Xiaoya Xiang}, {Chen Ding}, {Hao Luo}, {and} {Bin Bao}. 2013.
\newblock \showarticletitle{{HOTL}: a higher order theory of locality}. In {\em
  Proceedings of ASPLOS}. 343--356.
\newblock


\bibitem[\protect\citeauthoryear{Zhong, Dropsho, Shen, Studer, and Ding}{Zhong
  et~al\mbox{.}}{2007}]%
        {Zhong+:TOC07}
{Y. Zhong}, {S.~G. Dropsho}, {X. Shen}, {A. Studer}, {and} {C. Ding}. 2007.
\newblock \showarticletitle{Miss rate prediction across program inputs and
  cache configurations}.
\newblock {\it IEEE Trans. Comput.} {56}, 3 (March 2007), 328--343.
\newblock


\bibitem[\protect\citeauthoryear{Zhong, Shen, and Ding}{Zhong
  et~al\mbox{.}}{2009}]%
        {Zhong+:TOPLAS09}
{Yutao Zhong}, {Xipeng Shen}, {and} {Chen Ding}. 2009.
\newblock \showarticletitle{Program Locality Analysis Using Reuse Distance}.
\newblock {\em ACM TOPLAS\/} {31}, 6 (Aug. 2009), 1--39.
\newblock


\end{thebibliography}
\ULforem 	

\end{document}